\documentclass[aps,pra,floatfix,notitlepage,nofootinbib,longbibliography,11pt]{revtex4-1}
\pdfoutput=1
\usepackage{amsmath,amsthm,amsfonts,amssymb,braket}
\usepackage{graphicx}
\usepackage{color}

\makeatletter
\def\l@subsubsection#1#2{}
\makeatother

\usepackage[colorlinks=true,citecolor=blue,linkcolor=blue]{hyperref}
\usepackage[capitalize,nameinlink]{cleveref}

\newtheorem{lemma}{Lemma}[section]
\newtheorem{theorem}[lemma]{Theorem}
\newtheorem{conjecture}[lemma]{Conjecture}
\newtheorem{corollary}[lemma]{Corollary}

\theoremstyle{definition}
\newtheorem{definition}[lemma]{Definition}
\newtheorem{remark}[lemma]{Remark}

\newcommand{\CC}{{\mathbb{C}}}

\newcommand{\ZZ}{{\mathbb{Z}}}
\newcommand{\HH}{{\mathbb{H}}}
\newcommand{\FF}{{\mathbb{F}}}
\newcommand{\Id}{{\mathrm{Id}}}
\newcommand{\calC}{{\mathcal{C}}}

\newcommand{\SWAP}{{\mathrm{SWAP}}}

\newcommand{\lat}[2]{{{#1}\mathbb{Z}^{#2}}}
\newcommand{\half}{{\frac{1}{2}}}

\DeclareMathOperator*{\im}{{im}}
\DeclareMathOperator*{\diag}{{diag}}
\DeclareMathOperator*{\Supp}{{Supp}}
\DeclareMathOperator*{\rank}{{rank}}
\DeclareMathOperator*{\coker}{{coker}}
\DeclareMathOperator*{\Mat}{{Mat}}

\begin{document}
\title{Clifford Quantum Cellular Automata: \\Trivial group in 2D and Witt group in 3D}
\author{Jeongwan Haah}
\affiliation{Microsoft Quantum and Microsoft Research, Redmond, Washington, USA}
\email{jwhaah@microsoft.com}

\begin{abstract}
We study locality preserving automorphisms of operator algebras 
on $D$-dimensional uniform lattices of prime $p$-dimensional qudits (QCA),
specializing in those that are translation invariant (TI)
and map every prime $p$-dimensional Pauli matrix 
to a tensor product of Pauli matrices (Clifford).
We associate antihermitian forms of unit determinant over Laurent polynomial rings
to TI Clifford QCA with lattice boundaries,
and prove that the form determines the QCA up to Clifford circuits and shifts (trivial).
It follows that every 2D TI Clifford QCA is trivial 
since the antihermitian form in this case is always trivial.
Further, we prove that for any $D$ the fourth power of any TI Clifford QCA 
is trivial.
We present explicit examples of nontrivial TI Clifford QCA for $D=3$ and any odd prime~$p$,
and show that the Witt group of the finite field $\FF_p$ is a subgroup of 
the group $\mathfrak C(D = 3, p)$ of all TI Clifford QCA modulo trivial ones.
That is,
$\mathfrak C(D = 3, p \equiv 1 \mod 4) \supseteq \ZZ_2 \times \ZZ_2$ and
$\mathfrak C(D = 3, p \equiv 3 \mod 4) \supseteq \ZZ_4$.
The examples are found 
by disentangling the ground state of a commuting Pauli Hamiltonian
which is constructed by coupling layers of prime dimensional toric codes
such that an exposed surface has an anomalous topological order that
is not realizable by commuting Pauli Hamiltonians strictly in two dimensions.
In an appendix independent of the main body of the paper, 
we revisit a recent theorem of Freedman and Hastings
that any two-dimensional QCA, which is not necessarily Clifford or translation invariant,
is a constant depth quantum circuit followed by a shift.
We give a more direct proof of the theorem without using any ancillas.
\end{abstract}

\maketitle

\section{Introduction}

Quantum cellular automata (QCAs) have been studied in a variety of contexts including 
models of computation~\cite{Watrous1995} or algebra automorphisms~\cite{SchumacherWerner2004,GNVW,clifQCA,FreedmanHastings2019QCA}.
As an algebra automorphism,
a QCA has been identified as a disentangling scheme for a many-body state without any topological excitations~\cite{nta3}.
Ref.~\cite{SchumacherWerner2004} explains the distinction among various definitions.

We consider quantum cellular automata as automorphisms of $*$-algebra over $\CC$
generated by local operators of a lattice system of finite dimensional degrees of freedom (qudits).
Being an automorphism of a local algebra, not only does it preserve algebraic additions and multiplications
but also it maps a local operator to a local operator.
More precisely, we require that the image $\alpha(O)$ of any single site operator $O$ on $s$ by a QCA $\alpha$
be supported on a ball of radius $r < \infty$ centered at $s$.
Here the constant $r$, called the {\bf range} of $\alpha$,
is uniform across the whole lattice.
For example, a translation is a QCA;
however, a rotation is not a QCA in our definition
because the image of an operator $O$ that is far from the origin will be supported far from $O$.
Being a $*$-automorphism, the QCA that we consider is not an antiunitary.
Also, we do not consider systems with fermionic degrees of freedom, 
though it is an interesting setting on its own.

\subsection{Translation invariant Clifford QCAs}

The class of QCAs that we consider in this paper is even more specific.
Namely, we consider the $D$-dimensional hypercubic lattice $\ZZ^D$ of qudits of dimension $p$.
Each site of the lattice is occupied by finitely many qudits,
and the number of qudits per site must be uniform across the lattice.
That is, the local Hilbert space of the lattice is $(\CC^p)^{\otimes q}$ everywhere.
We denote such a system by~$\lat{q}{D}$.
The $*$-algebra of local operators in this lattice is generated%
\footnote{
Our usage of the verb ``generate'' follows that of algebra.
The generated algebra is the set of all \emph{finite} linear combinations over $\CC$ of
\emph{finite} product of generators.
There are many ways to complete this algebra under various operator topologies;
however, we will not consider any completion 
since we will handle tensor products of generalized Pauli matrices only,
but no linear combinations thereof.
}
over $\CC$ by
{\bf generalized Pauli operators} $X$ and $Z$ on each site where
\begin{align}
X = \sum_{j=0}^{p-1} \ket {j+1 \mod p}\bra j, \quad Z = \sum_{j=0}^{p-1} \exp(-2\pi i j/p) \ket j \bra j, 
\quad XZ = e^{2\pi i /p} ZX . \label{eq:pauli}
\end{align}
We are interested in {\bf Clifford QCAs} that map every generalized Pauli matrix
to a tensor product of generalized Pauli matrices.
We further demand that our Clifford QCA is \emph{translation invariant} (TI);
our Clifford QCA is assumed to commute with any translation QCA.
Since our translation group of the lattice is an abelian group $\ZZ^D$,
any translation QCA ({\bf shift}) is a TI Clifford QCA.
The set of all TI Clifford QCAs form a group under the composition;
the product $\alpha \beta $ of two QCAs of ranges $r$ and $r'$, respectively,
is a QCA of range $r + r'$,
and $\alpha^{-1}$ exists since $\alpha$ is an automorphism 
and is a QCA of range $r$~\cite{SchumacherWerner2004,Arrighi2007}.

An obvious subgroup of the group of all TI Clifford QCAs is generated by
layers of Clifford gates
whose geometric arrangement obeys the translation invariance.
Here, within a layer of Clifford gates,
every gate is supported on qudits within a ball of radius $r$,
and every gate either has support that is disjoint from that of any other gate
or is commuting with any other gate.
Thus, a {\bf TI Clifford circuit} is a finite composition of such layers,
and the number of layers is the {\bf depth} of the circuit.
The sum of all the locality bound $r_j$ in each layer $j$ of a circuit
bounds the range of the circuit from above.

\subsection{Equivalent QCAs and the group in question}

We will sometimes weakly break the translation symmetry of our system 
so that the new translation symmetry group is a subgroup of $\ZZ^D$ of finite index.
This typically arises when we wish to compose a QCA $\alpha$ on a system $A = \lat{q}{D}$
with another QCA $\beta$ on, say, a system $B = \lat{2q}{D}$.
In that case, we may combine two sites of $A$ to regard them as a new bigger site,
effectively forgetting the finest translation structure in $A$.
We will be explicit in statements whenever we reduce the translation group down to a subgroup,
but we will not be interested in the index of the subgroup, other than that it is finite.
We will use the term {\bf weak} translation invariance
when we speak of such weak translation symmetry breaking.

For two QCAs on the same dimensional but disjoint systems
we can take the tensor product $\alpha \otimes \beta$ in the obvious way.
\begin{definition}\label{def:equivQCA}
Let $\alpha$ be a TI Clifford QCA on the system $\lat{q}{D}$,
and $\beta$ be another TI Clifford QCA on $\lat{q'}{D}$.
We say $\alpha$ and $\beta$ are {\bf equivalent} and write $\alpha \cong \beta$
if $\alpha \otimes \beta^{-1} \otimes \Id$
equals a translation QCA followed by a Clifford circuit that is weakly translation invariant,
where $\Id$ is the identity QCA on $\lat{q''}{D}$ for some $q'' \ge 0$.
We denote by $\mathfrak C(D,p)$ the set of all equivalence classes of TI Clifford QCAs 
on the uniform lattice $\ZZ^D$ with finitely many $p$-dimensional qudits per site.
\end{definition}
This is an equivalence relation 
because $\alpha \otimes \alpha^{-1}$ is a Clifford circuit for any Clifford QCA~$\alpha$~\cite{Arrighi2007}.
Thus, we treat any Clifford circuit or a translation QCA as {\bf trivial}.
The appearance of $\Id$ is to implement so-called {\bf stabilization} in the equivalence.%
\footnote{
By a recent result on ``ancilla removal''~\cite{FHH2019},
we only need weak translation symmetry breaking.
That is,
the equivalence relation remains the same with or without the stabilization.
}
This is motivated by a practice in topological quantum phases of matter
where one is allowed to add and subtract unentangled degrees of freedom ({\bf ancilla})
to relate a state to another.
For each $D$ and $p$ the set $\mathfrak C(D,p)$ is obviously an abelian group
under the tensor product.
This group is isomorphic to the group of QCAs under composition~\cite[Lem.~I.7]{nta3}.
It is our main problem to identify $\mathfrak C (D,p)$.

Let us summarize known results.
Trivially, we see $\mathfrak C (D=0,p) = 0$.
It has been known that $\mathfrak C(D=1,p) = 0$ when $p$ is a prime:
Ref.~\cite{clifQCA} classifies TI Clifford QCAs on $\lat{1}{1}$,
Ref.~\cite{GNVW} classifies general 1D QCAs under a stronger notion of equivalence,
and Ref.~\cite{Haah2013} independently computes $\mathfrak C(D=1,p) = 0$ for any prime $p$.
For $D=2$, it is shown~\cite{FreedmanHastings2019QCA} 
that any QCA (not necessarily Clifford or TI)
is a constant depth quantum circuit (that is not necessarily Clifford) plus shifts with ancillas;
see \cref{app:trivialQCA}.
However, in higher dimensions no definite answer is known.
For $D=3$, it is shown~\cite[Thm.III.16]{nta3} that $\mathfrak C(D=3, p=2) \neq 0$.
For any $D \ge 0$, it is shown~\cite[Thm. IV.9]{nta3} 
that $\alpha^2 \cong \Id$ for any $[\alpha] \in \mathfrak C(D,p=2)$.

The result of this paper is the following.
\begin{theorem}\label{thm:main}
Let $D \ge 0$ be an integer, and $p$ a prime.
Let $\alpha$ be any TI Clifford QCA on $\lat q D$ of $p$-dimensional qudits.
\begin{itemize}
\item[(i)]
If $D=2$, $\alpha$ is a weakly TI Clifford circuit followed by a shift. Hence, $\mathfrak C(D=2,p) = 0$.
\item[(ii)] 
Suppose $p \equiv 1 \mod 4$. 
Then, $\alpha^{\otimes 2}$ is a weakly TI Clifford circuit followed by a shift.
Hence, $\mathfrak C(D,p)$ is a direct sum of some number of $\ZZ_2$.
In fact, $\mathfrak C(D=3,p) \supseteq \ZZ_2 \times \ZZ_2$.
\item[(iii)] 
Suppose $p \equiv 3 \mod 4$. Then, $\alpha^{\otimes 4}$ is a weakly TI Clifford circuit followed by a shift.
Hence, $\mathfrak C(D,p)$ is a direct sum of some number of $\ZZ_4$ and $\ZZ_2$.
In fact, $\mathfrak C(D=3,p) \supseteq \ZZ_4$.
\end{itemize}
\end{theorem}
\begin{remark}
The direct sum could be empty (as in $D=2$) or infinite.
We do not know whether the groups are always finite.
$\square$
\end{remark}

The dimension $p$ of a qudit will always be a {\em prime}.
A {\bf Pauli operator} is a finite tensor product of generalized Pauli matrices and a $p$-th root of unity;
we usually drop the word ``generalized.''
We use modules over Laurent polynomial rings.
By convention, any map on a module (e.g., a matrix acts on a free module) acts on the left.
All algebras and subalgebras are unital and over $\mathbb C$.

\subsection{Overview of the proof}

We calculate the exponent of TI Clifford QCA in \cref{sec:exponents}
by mimicking the Witt group calculation for finite fields
together with Suslin's stability theorem~\cite{Suslin1977Stability}.
In \cref{sec:bdalg}
we establish a circular correspondence among three objects:
a TI Clifford QCA $\alpha$ on $D$ dimensions, 
a subalgebra on a $(D-1)$-dimensional sublattice
in the image under $\alpha$ of operators on a half space,
and a nondegenerate antihermitian form over Laurent polynomial ring of $D-1$ variables.
The lost information in this circular correspondence will be precisely Clifford circuits and shifts.
By studying antihermitian forms over univariate Laurent polynomial rings,
we conclude that all $D=2$ TI Clifford QCA must be trivial.
This triviality result is constructive;
the proof will reveal an algorithm how to decompose a TI Clifford QCA into a Clifford circuit and a shift.
In general, the antihermitian forms give ``Witt groups.''

The nontriviality results for $D=3$
are based on examples of coupled layer constructions~\cite{WangSenthil2013}.
All of our results share a common theme 
that we examine an associated lattice Hamiltonian
(i.e., a separator~\cite{nta3}) with a spatial boundary.

The rest of the paper, except for \cref{sec:discussion} and \cref{app:trivialQCA}, 
constitutes the proof of \cref{thm:main}.

\section{Exponents of Clifford QCA}\label{sec:exponents}

An {\bf exponent} of a group $G$ is an integer $n$ 
such that $g^n$ is the identity for every $g \in G$;
it does not always exist.
In this section we show that the group of 
all translation invariant Clifford QCAs modulo Clifford circuits and shifts
on any lattice of prime dimensional qudits
has an exponent 2 or 4.

\subsection{Translation invariant Clifford QCAs and polynomial symplectic matrix}

A Clifford QCA on $\lat{q}{D}$ is determined by the images of Pauli matrices.
Using translation invariance, 
we have developed~\cite{Haah2013} a compact representation of the data
by a {\bf symplectic matrix}%
\footnote{
The term ``symplectic'' is overloaded from the situation where $D=0$.
Our symplectic matrix is an autormophism of a free module over the Laurent polynomial ring with coefficients in a finite field,
preserving an ``antihermian'' pairing (form).
So, we should have called our matrix ``unitary'' following literature on quadratic forms~\cite{MilnorHusemoller,Lam}.
However, there is a greater danger of confusion if we chose ``unitary'' instead of ``symplectic,''
since we consider complex inner product spaces and unitary operators on them, too.
}
over $R = \ZZ_p[\ZZ^D] = \FF_p[x_1^\pm,\ldots,x_D^\pm]$; 
see an introductory section~\cite[IV.A]{nta3}.
This is a generalization of a machinery 
behind Pauli stabilizer block codes~\cite{CalderbankRainsShorEtAl1997Quantum}
to infinite lattices.
The developed machinery associates
a matrix $Q$ over $R$ to any TI Clifford QCA $\alpha$ such that
\begin{align}
Q^\dagger \lambda_q Q = \lambda_q := 
\begin{pmatrix}
0 & I_q \\
- I_q & 0
\end{pmatrix}.
\end{align}
Here, $\dagger$ is the transpose followed by the {\bf antipode map}, a $\FF_p$-linear involution
\begin{align}
x_j \mapsto \bar x_j = x_j^{-1}
\end{align}
of~$R$, corresponding to the inversion of the lattice $\ZZ^D$ about its origin.%
\footnote{
If $R$ is regarded as a group algebra over $\FF_p$,
this involution is the usual antipode map of the group algebra.
For the group algebra for a nonabelian group, the antipode map is not an automorphism 
because it reverses the multiplication order.
But our group $\ZZ^D$ is abelian, and therefore the antipode map is an $\FF_p$-algebra automorphism of $R$.
}
The set of all such matrices $Q$ maps into $\mathfrak C(D,p)$,
where the map becomes surjective 
if we enlarge the domain by taking the union of such sets of matrices over all~$q$.
The tensor product of QCAs corresponds to a direct sum of the corresponding symplectic matrices.
The lost information in the association from $\alpha$ to $Q$
is only a nonentangling Clifford circuit (actually consisting of Pauli matrices)
of depth one~\cite[Prop. 2.1]{Haah2013}.

We will frequently use matrix operations involving the antipode map,
which we denote as follows. 
For any matrix $M$ let $M_{\mu \nu}$ denote its matrix element.
\begin{align}
(M^\dagger)_{\mu \nu} &= \overline{M_{\nu \mu}}\nonumber \\
M^{-\dagger} &= (M^{-1})^\dagger = (M^\dagger)^{-1}
\end{align}
These notations will be used for any invertible element of $R$ which is a $1$-by-$1$ matrix.

The generators for Clifford circuits correspond to elementary row operations that are symplectic~\cite{Haah2013}.
Define
\begin{align}
 \left[ E_{i,j}(u) \right]_{\mu \nu} &= \delta_{\mu \nu} + \delta_{\mu i} \delta_{\nu j} u 
 & \text{ where $\delta$ is the Kronecker delta.}
\end{align} 
The elementary row operations corresponding to TI Clifford circuit and shift layers are:
\begin{align}
 \text{Hadamard: } \quad&E_{i,i+q}(-1) E_{i+q,i}(1) E_{i,i+q}(-1) &\text{ where } 1 \le i \le q,\nonumber\\
 \text{Control-Phase: }\quad & E_{i+q,i}(f)&\text{ where }f = \bar f \in R, 1 \le i \le q,\label{eq:ElemSymp}\\
 \text{Control-X: } \quad & E_{i,j}(a) E_{j+q,i+q}(-\bar a) &\text{ where } a \in R, 1 \le i \ne j \le q, \nonumber \\
\text{Extra gate $J$ for }p \neq 2: \quad & E_{i,i}(c-1) E_{i+q,i+q}(c^{-\dagger}-1) & \text{ where } c \in R^\times, 1 \le i \le q. \nonumber
\end{align}
Here, $c$ in the extra gate $J$ belongs to $\FF_p$ if it represents a single qudit gate,
but $c$ is a monomial $x_1^{e_1} \cdots x_D^{e_D}$ with at least one $e_j$ nonzero
if $J$ represents a nonidentity shift QCA.
The question of whether a given TI Clifford QCA $\alpha$ is a Clifford circuit and a shift
boils down to whether the symplectic matrix $\alpha$ is a product of these {\bf elementary symplectic operations}.
Note that the first three operations all have determinant one, 
whereas the extra gate $J$ has determinant $c c^{-\dagger}$ that is~$1$ if and only if~$c \in \FF_p^\times$.
The superscript `$\times$' means the multiplicative group of all invertible elements.
\begin{lemma}\label{lem:detQ}
The determinant of any symplectic matrix $Q$ is $c c^{-\dagger}$ for some $c \in R^\times$.
\end{lemma}
\begin{proof}
Suppose the lemma is true for $D \le 1$.
Then, for $D > 1$, if $\det Q = a x_1^{e_1} \cdots x_D^{e_D}$ for some $a \in \FF_p$ and $e_j \in \ZZ$,
we may set all but one $x_j$ to be~$1$, and conclude that $a = 1$ and $e_j \in 2\ZZ$ by the result of $D=1$ case.
Hence, we only have to show the lemma for $D \le 1$.
In this case, we claim that $Q$ is a product of elementary symplectic matrices.
Then, $\det Q$ is a product of the determinants of elementary symplectic matrices,
each of which is $c c^{-\dagger}$ for some $c \in \FF_p[x^\pm]^\times$,
and hence the lemma will be proved.
The decomposition of $Q$ into elementary symplectic matrices when $D=1$ 
is a consequence of \cite[\S 6]{Haah2013} and \cite[Lem.IV.10]{nta3},
but we reproduce the detail here.

The entries of the first column $v$ of $Q$ generate the unit ideal (unimodular), 
as seen by the equation $(Q^\dagger \lambda) Q = \lambda$
where the first column of $\lambda$ has an entry $1$ at $q+1$-th position.
Moreover, $v^\dagger \lambda v = 0$ (self-orthogonal).
We claim that (i) for any unimodular self-orthogonal column matrix $v$
there exists a product $E$ of the elementary symplectic matrices 
such that $Ev$ is a column matrix with a unique nonzero entry $1$ at the first position,
and moreover 
(ii) if in addition $v$ has $1$ at the first position, 
$E$ may be chosen such that $EQ$ and $Q$ have the same first row.
Then, by (i) we find an elementary symplectic $E$ 
such that $EQ$'s first column is the unit vector with $1$ at the first position.
For any matrix $P$, the equation $P^\dagger \lambda P = \lambda$ 
implies that $(P^\dagger \lambda) (-P\lambda) = I$, so $(-P\lambda) (P^\dagger \lambda) = I$,
which means $P \lambda P^\dagger = \lambda$.
Hence, we see that the first \emph{row} of $EQ$ is unimodular and self-orthogonal,
and has $1$ at the first position.
By applying (ii) to rows, 
we find an elementary symplectic $E'$ such that
\begin{align}
EQE'= \begin{pmatrix}
1 & 0_{1 \times (q-1)} & 0 & 0_{1 \times (q-1)} \\
0_{(q-1)\times 1} & \star & \star & \star \\
0 & \star & \star & \star \\
0_{(q-1)\times 1} & \star & \star & \star
\end{pmatrix}
=
\begin{pmatrix}
1 & 0_{1 \times (q-1)} & 0 & 0_{1 \times (q-1)} \\
0_{(q-1)\times 1} & \star & 0 & \star \\
0 & 0 & 1 & 0 \\
0_{(q-1)\times 1} & \star & 0 & \star
\end{pmatrix}
\end{align}
where the second equality is by the equation $(EQE')^\dagger \lambda (EQE') = \lambda$.
This allows induction in $q$, completing the decomposition of $Q$ into elementary symplectic matrices.

It remains to prove the two claims (i) and (ii).
Control-X gates induce arbitrary row operations on the upper half of $v$.
Euclid's algorithm implies that this upper half of $v$ can be transformed by, say, $E_1$,
such that there remains in $E_1 v$ a sole nonzero entry (at the first position)
that is the greatest common divisor of the entries of the upper half of $v$.
If the first position of $v$ is $1$, $E_1$ can be chosen such that
the first row of $E_1$ is the same as that of the identity matrix.
Once this is done, we can focus on the lower half of $E_1 v$,
and Control-X gates focusing on the $q+2,\ldots,2q$-th entries can induce $E_2$ 
such that nonzero entries of $E_2 E_1 v$ are now only at $1,q+1,q+2$-th positions.
Hence, we have reduced the problem to $q=2$ case.

If the $q+2$-th entry is nonzero and if the first position of $v$ is $1$,
we can use Hadamard to bring it to the $2$nd position,
and eliminate it by Control-X.
Since the first position of $v$ is $1$,
the $q+1$-th entry $f$ must be equal to $\bar f$ due to the self-orthogonality,
and Control-Phase can eliminate $f$.
This completes the proof of the claim (ii).

For (i), consider the degrees of nonzero entries of $v$ which are Laurent polynomials.
Here, the relevant degree is not the absolute exponent, 
but the difference between the maximum exponent and the minimum.
For example, the relevant degree of $x^{-3} + x + x^2$ is $5 = 2 - (-3)$.
The minimum relevant degree over all nonzero entries of $E v$ for an elementary symplectic $E$,
must reach $0$ for some $E$ that is a product of Control-X and Hadamard, since $v$ is unimodular.
The extra gate $J$ turns an entry of relevant degree~$0$ into~$1$,
and the claim (ii) completes the proof. 
\end{proof}

\subsection{To Clifford circuits}

Suslin's stability theorem~\cite[Cor.7.11]{Suslin1977Stability}
says that any invertible matrix over $R$ of unit determinant and of size $\ge 3$ 
is a product of elementary row operations.
Hence, for any $E, F \in GL(n \ge 3; R)$ the following transformation is
by a circuit on the left of $\alpha$ and another circuit on the right of $\alpha$,
possibly with a shift QCA.
\begin{align}
Q = \begin{pmatrix}
A & B \\ C & D
\end{pmatrix}
\to
\begin{pmatrix}
E & 0 \\ 0 & E^{-\dagger}
\end{pmatrix}
\begin{pmatrix}
A & B \\ C & D
\end{pmatrix}
\begin{pmatrix}
F & 0 \\ 0 & F^{-\dagger}
\end{pmatrix}
=
\begin{pmatrix}
EAF & E B F^{-\dagger} \\ E^{-\dagger}CF & E^{-\dagger}DF^{-\dagger}
\end{pmatrix}. \label{eq:suslin}
\end{align}
If $E$ and $F$ are some product of elementary circuits as listed in \cref{eq:ElemSymp},
then of course we do not need the condition $n \ge 3$.
Even in a nontrivial application of Suslin's theorem,
the condition $n \ge 3$ is not a restriction 
since we may simply add one ancilla qudit per lattice point of $\ZZ^D$
or enlarge the unit cell by reducing translation group down to a subgroup of index, say, $2$.

The calculation here is similar to that of \cite[Thm. IV.9]{nta3}.
It suffices to trivialize the left half block of a symplectic matrix $\alpha$
since the other half is determined up to a TI Clifford circuit~\cite[Lem.IV.10]{nta3}.

First, let us treat the case where $p \equiv 1 \mod 4$.
In this case, $(p-1)/4$ is an integer.
Hence, if $g \in \FF_p^\times$ is the generator 
of the multiplicative group $\FF_p^\times = \FF_p \setminus \{0\}$ that is cyclic,
$g^{(p-1)/4}$ is a primitive fourth root of unity,
and in particular $\sqrt{-1}$ exists in $\FF_p$.
Therefore, we may choose $E^{-\dagger} = F = \sqrt{-1} I$ (so $E = F^{-\dagger} = - \sqrt{-1}I$)
in \cref{eq:suslin} to see
\begin{align}
Q
=
\begin{pmatrix}
A & B \\
C & D \\
\end{pmatrix}
\cong
\begin{pmatrix}
A & -B \\
-C & D \\
\end{pmatrix}. \label{eq:p1mod4symp}
\end{align}
That is, $\alpha \cong \alpha'$ 
where the symplectic matrix of $\alpha'$ 
is the right-hand side of \cref{eq:p1mod4symp}.
The next lemma shows that $\alpha' \cong \alpha^{-1}$,
which implies $\alpha^2 \cong \Id$ if $p \equiv 1 \mod 4$,
proving the first statement of \cref{thm:main}~(ii).
\begin{lemma}\label{lem:inverseQCA}
Let $p = 2,3,5,7,\ldots$ be any prime and $D\ge 0$ an integer,
Let
\begin{align}
Q = \begin{pmatrix}
A & \star \\ C & \star
\end{pmatrix}, \quad 
Q' = \begin{pmatrix}
A & \star \\ -C & \star
\end{pmatrix}
\end{align}
be symplectic matrices.
Then, $Q \oplus Q'$ is a product of elementary row operations corresponding to TI Clifford circuit layers.
That is, $[Q'] \in \mathfrak C(D,p)$ is the inverse of $[Q] \in \mathfrak C(D,p)$.%
\footnote{
In our convention, the upper half block of a symplectic matrix is for tensor components of $X$
that is a real matrix,
and the lower is for those of $Z$ that is a nonreal matrix.
Hence, the separator associated with $Q'$ (the left half block of $Q'$)
may be regarded as the time reversal conjugate of
the separator associated with $Q$.
}
\end{lemma}
This can be thought of as a generalization of \cite[Thm.IV.9]{nta3} that asserts $\mathfrak C(D,p=2)$ 
has exponent~$2$, to any prime~$p$;
with $p=2$ the matrix $Q'$ can be set to be equal to $Q$.
\begin{proof}
We use \cref{eq:suslin} with various $E$ and $F$ on the full matrix $S$
as well as some elementary Clifford gates. With the right half block suppressed, the result is
\begin{align}
\begin{pmatrix}
A & 0 \\
0 & A \\
-C & 0\\
0 & C 
\end{pmatrix}
\to
\begin{pmatrix}
A & 0 \\
0 & A \\
0 & -C\\
-C & C
\end{pmatrix} 
\to 
\begin{pmatrix}
A & 0 \\
C & -C \\
0 & -C  \\
0 & A 
\end{pmatrix}
\to
\begin{pmatrix}
I & K  \\
0 & K'  \\
0 & 0  \\
0 & I 
\end{pmatrix}
\to
\begin{pmatrix}
I & 0 \\
0 & 0 \\
0 & 0 \\
0 & I  
\end{pmatrix}
\to
\begin{pmatrix}
I & 0 \\
0 & I \\
0 & 0 \\
0 & 0 
\end{pmatrix}.
\end{align}
The first arrow is by \cref{eq:suslin} with
$E = \begin{pmatrix} I & 0 \\ I & I \end{pmatrix} = F^{-1}$.
The second arrow is by the Hadamard gate on every qudit on the second system.
The third arrow is by \cref{eq:suslin} with 
$E = Q^{-1}$ and $F = I$.
It is this step where the Suslin theorem~\cite{Suslin1977Stability} is crucial.
The fourth arrow is by \cref{eq:suslin} with $E = I$ and 
$F = \begin{pmatrix} I & -K \\ 0 & I \end{pmatrix}$,
followed by \cite[Lem.IV.10]{nta3} on the second system which eliminates $K'$.
The fifth arrow is by the inverse Hadamard gate on every qudit on the second system.
\end{proof}

Next, we treat the case where $p \equiv 3 \mod 4$.
In this case, $(p-1)/2$ is an odd integer, and $-1 \in \FF_p$ is not a square;
i.e., the equation $x^2 + 1 = 0$ has no solution in $\FF_p$.
But still the equation $x^2 + y^2 + 1 = 0$ has a solution $x, y \in \FF_p$.
For example, $5^2 + 3^2 + 1 = 0$ in $\FF_7$.
(This elementary fact can be shown in many ways:
one is to count $\{ x^2 \in \FF_p: x=0,1,2,\ldots,(p-1)/2 \}$, and
another is to consider the quotient multiplicative group $\FF_p^\times / (\FF_p^\times)^2$
and observe that $(\FF_p^\times)^2$ cannot be closed under additions when $-1 \notin (\FF_p^\times)^2$.)
Using $x,y$ such that $x^2 + y^2 = -1$,
we set $E^{-\dagger} = F = \begin{pmatrix} x I & y I \\ y I & -x I \end{pmatrix}$ in \cref{eq:suslin}
with $Q \oplus Q$ in place of $Q$:
\begin{align}
Q \oplus Q 
=\begin{pmatrix}
A & 0 & B & 0 \\
0 & A & 0 & B \\
C & 0 & D & 0 \\
0 & C & 0 & D
\end{pmatrix}
\cong
\begin{pmatrix}
A & 0 & -B & 0 \\
0 & A & 0 & -B \\
-C & 0 & D & 0 \\
0 & -C & 0 & D
\end{pmatrix}
\cong Q^{-1} \oplus Q^{-1}. \label{eq:qq=qqinv}
\end{align}
where the second equivalence is by \cref{lem:inverseQCA}.
This implies that
\begin{align}
Q^{\oplus 4} \cong Q^{\oplus 2} \oplus (Q^{-1})^{\oplus 2} \cong \Id .
\end{align}
We have proved the first statement of \cref{thm:main}~(iii).

The Witt group of a field of characteristic not $2$
is a group of all equivalent nondegenerate symmetric forms up to hyperbolic forms
$\cong \diag(1,-1)$.
For finite fields of $p$ elements,
the Witt group is isomorphic to $\ZZ_2 \times \ZZ_2$ if $p \equiv 1 \mod 4$
and $\ZZ_4$ if $p \equiv 3 \mod 4$;
see a general book~\cite{Lam} or an elementary section \cite[App.E]{HHPW2017}.
Indeed, our computation here is motivated by that of the exponent computation of the Witt group of $\FF_p$.
The connection will be tighter in the $D=3$ examples of TI Clifford QCAs in \cref{sec:3DQCA}.

\section{Boundary algebras}\label{sec:bdalg}

At the end of this section we will prove the triviality of 2D TI Clifford QCA.
To this end we will examine the ``boundary'' of a QCA,
which is basically the image of the subalgebra of all operators supported on a half space.
This image contains all local operators far from the boundary of the half space,
but becomes nontrivial as we get closer to the boundary.
Since we study QCA of strictly finite range,
the region where the image is nontrivial is sharply contained in a finite neighborhood of the boundary.
This approach has proven useful in previous classification results~\cite{GNVW,FreedmanHastings2019QCA}.

Since we study translation invariant Clifford QCA,
where the objects are presumably simpler,
our characterization of boundary algebras will be more thorough.
We will identify generators of a boundary algebra,
define an antihermitian form over the translation group algebra along the boundary,
and consider the collection of all such forms.
They will form a group in an analogous way as the Witt groups for fields.
In the converse direction,
we will construct a QCA starting from any invertible antihermitian forms,
completing the characterization of the boundary algebras.
All these results will have no dependence in the spatial dimension $D$.

The correspondence among three classes of objects, TI Clifford QCA, boundary algebras, and invertible antihermitian forms,
will be so tight that the only lost information
turns out to be precisely the Clifford circuits and shifts.
It is very useful for us that the boundary algebras and the antihermitian forms
live in one dimension lower than the QCA.
In particular, the problem on 2D TI Clifford QCA becomes 
that on quadratic forms over a field of transcendence degree~$1$ above a finite field,
and we deduce that any such forms are trivial.

Let us begin with a definition of boundary algebras.
\begin{definition}
Let $Q$ be the symplectic matrix of a TI Clifford QCA $\alpha$ on $\lat q D$ where $D \ge 1$.
Suppose that the variable $x_D$ in $Q$ has exponent $0$ or $1$,
so that $Q = A + x_D B$ where $A$ and $B$ are matrices over $\FF_p[x_1^\pm,\ldots,x_{D-1}^{\pm}]$.
We define the {\bf boundary algebra} of $\alpha$ (at positive $x_D$-axis) to be
one that generated by Pauli operators corresponding to the columns of $B$.
\end{definition}
The columns of $B$ specify the Pauli operators only up to phase factors,
but this phase ambiguity is unimportant as we think of an algebra generated by them.

The boundary algebra for more general QCA $\alpha$ 
has been examined before~\cite{GNVW,FreedmanHastings2019QCA}
and is defined as the ``difference'' between $M$ and $M'$ 
where $M$ the image under $\alpha$ of full algebra on a half space~$\HH$
and $M'$ is the full algebra of a (smaller) half space whose image under $\alpha$ is contained in $\HH$.
Our definition is tailored to the special case of translation invariant Clifford QCA.

However, our definition is somewhat unsatisfactory
in the following sense.
We put a restriction that the variable $x_D$ 
must have exponent either~$0$ or~$1$.
This condition can be met by 
composing a given TI Clifford QCA with shifts (that we regard trivial)
so that the symplectic representation does not have any variable with negative exponents,
and taking a translation subgroup of some finite index so that
the symplectic representation has exponents $0,1$ for all variables.
Then, it becomes necessary to show that the boundary algebras
resulting from arbitrary choices of translation subgroup and shifts,
are the same in an appropriate sense.
We will give some answers in this direction later in \cref{lem:CircuitsGiveTrivialAntihermitianForms}.

\subsection{Commutation relations among generators}\label{sec:constructXi}

In this subsection $Q = A + x_D B$ is a $2q \times 2q$ symplectic matrix over 
$\FF[x_1^\pm,\ldots,x_{D}^{\pm}]$ 
satisfying $Q^\dagger \lambda Q = \lambda$
where $A,B$ are matrices over $\FF[x_1^\pm,\ldots,x_{D-1}^{\pm}]$.
We suppress the dimension index $q$ from the antisymmetric matrix $\lambda = \lambda_q$
which will be inferred from the context.
Since $Q$ is invertible, its determinant is a unit in $\FF[x_1^\pm,\ldots,x_{D}^{\pm}]$, a monomial.
Except for \cref{cor:bdalgComm} below, all the lemmas in this subsection are true 
for any (not necessarily finite) field $\FF$.

\begin{lemma}\label{lem:bdalg}
The ranks of $A$ and $B$ are both even and sum to $2q$, 
and their smallest nonvanishing determinantal ideals are both unit.
Furthermore, $\ker A^\dagger \lambda = \im B$ and $\ker B^\dagger \lambda = \im A$
over $\FF[x_1^\pm,\ldots,x_{D-1}^{\pm}]$.
\end{lemma}
\begin{proof}
Put $z = x_D$ so that
$\det Q = x_1^{e_1} \cdots x_{D-1}^{e_{D-1}} z^\ell$ for $e_j,\ell \in 2\ZZ$ by \cref{lem:detQ}.
For any $k=1,\ldots,2q$, let $A_k$ and $B_k$ be the $k$-th column of $A$ and $B$, respectively.
Since the determinant is multilinear in columns, we expand $\det Q$ in $z$ as
\begin{align*}
\det Q 
&=
\det ( A_1 + z B_1 | \cdots | A_{2q} + z B_{2q} )\\
&=
z^0 \det ( A_1 | \cdots | A_{2q} ) \\
&\quad + \cdots  \\
&\quad + z^\ell \Big[ \det (A_1 | \cdots | A_{2q-\ell}| B_{2q-\ell+1} | \cdots | B_{2q} ) + \cdots + 
\det(B_1| \cdots | B_\ell| A_{\ell+1} | \cdots |A_{2q}) \Big] \\
&\quad + \cdots \\
&\quad + z^{2q}\det ( B_1 | \cdots | B_{2q} ) .
\end{align*}
The $z^\ell$ term must be equal to $\det Q$, and hence all the other terms vanish.
The $z^\ell$ term is a sum of $\binom{2q}{\ell}$ determinants,
each of which contains $\ell$ columns of $B$ and $2q-\ell$ columns of $A$.
Now, the cofactor expansion formula for the determinant implies that
\begin{align}
\det ( A_{j_1} | \cdots |A_{j_m} | B_{j_{m+1}} | \cdots | B_{j_{2q}} ) \in I_m(A) \cap I_{2q-m} (B)
\end{align}
for any $m$ distinct columns of $A$ and $2q-m$ distinct columns of $B$
where $I_k(M)$ is the $k$-th determinantal ideal of $M$,
i.e., the ideal generated by all $k \times k$ minors of $M$.
Therefore, the $z^\ell$ term of $\det Q$, which is a unit,
is in $I_{2q - \ell}(A) \cap I_\ell(B)$.
It follows that $I_{2q-\ell}(A) = I_\ell(B) = (1)$, the unit ideal.
In particular, the sum of ranks of $A$ and $B$ is $\ge 2q$.

Expanding the condition $Q^\dagger \lambda Q = \lambda$ in $z$,
we see that $A^\dagger \lambda A + B^\dagger \lambda B = \lambda$ and $A^\dagger \lambda B = 0$.
The latter implies that the sum of the ranks of $A$ and $B$ is $\le 2q$.
Therefore, the rank of $A$ is $2q - \ell$ and that of $B$ is $\ell$,
and the first claim is proved.

The exactness claim follows from \cite[Lem.20.10]{Eisenbud}.
\end{proof}

Since $B$ has the unit determinantal ideal, 
we may choose a free basis for $\im B$ by the following.
\begin{lemma}\label{lem:Quillen-Suslin-Swan}
Let $M$ be an $n \times m$ matrix over $R = \FF[x_1 ^\pm, \ldots, x_D^\pm]$ where $\FF$ is any field.
If the smallest nonzero determinantal ideal $I(M)$ of $M$ is unit,
then there exists invertible matrices $E,E'$ such that
\begin{align}
E M E' = \begin{pmatrix} I_{\rank M} & 0 \\ 0 & 0 \end{pmatrix}.
\end{align}
\end{lemma}
This is a direct consequence of the Quillen-Suslin-Swan theorem~\cite{Suslin1977Stability,Swan1978}.
This formulation has been used in \cite{nta3}, but here we spell out details.
\begin{proof}
Upon localization at any maximal ideal of $R$,
the claim is obvious by elementary row and column operations.
Thus, $\coker M$ is locally free, 
and hence is projective~\cite[I.3.4]{LamSerreBook}.
The Quillen-Suslin-Swan theorem says that $\coker M$ is actually free.
Put $\coker M \cong R^k$.
Composing this isomorphism with the projection from $R^n$ onto $\coker M$,
we have an exact sequence
\begin{align}
R^m \xrightarrow{M} R^n \xrightarrow{N} R^k \to 0.
\end{align}
In particular, each row of $N$ generates the unit ideal (unimodular)
and \cite[Thm.7.2]{Suslin1977Stability} implies that any row can be brought
to the unit row vector with sole nonzero entry~$1$ by some basis change of~$R^n$,
which in turn can be used to diagonalize $N$ by some basis change of~$R^k$.
(The cited theorem does not cover cases where $n \le 2$ but in those cases this claim is trivial.)
Once $N$ is diagonalized with the identity matrix as a principal matrix,
the map $M$ has to map onto the components of $R^n$ corresponding to zero columns of $N$:
\begin{align}
N = \begin{pmatrix} I & 0 \end{pmatrix}, 
\quad M = \begin{pmatrix} 0 \\ M_0 \end{pmatrix} .
\end{align}
Here, $M_0$ is now onto, so a similar argument shows that $M_0$ can be diagonalized.
\end{proof}

Let $B_0$ be a matrix whose columns form a free basis of $\im B$,
as provided by \cref{lem:Quillen-Suslin-Swan}.
\begin{lemma}\label{lem:detXi}
Put $\Xi = B_0^\dagger \lambda B_0 = - \Xi^\dagger$.
Then, $\det \Xi = c^2$ for some $c \in \FF^\times$.
\end{lemma}
\begin{proof}
Let us first show that $\Xi = B_0^\dagger \lambda B_0$ is invertible.
\cref{lem:Quillen-Suslin-Swan} says that there exists a matrix $B'$ such that $B_0^\dagger B' = I_{\rank B}$.
Since $\lambda$ is invertible, we see that there exists a column matrix $v_j$ 
such that $B_0^\dagger \lambda v_j$ is the unit column matrix with a sole nonzero entry at the $j$-th component,
for any $j = 1,\ldots, \rank B_0$.
We claim that there exists $v_j'' \in \im B_0$ such that $B_0^\dagger \lambda v_j'' = B_0^\dagger \lambda v_j$.
Since $j$ is arbitrary, this shows that
$\Xi : \FF[x_1^\pm, \ldots, x_{D-1}^\pm]^\ell \to \FF[x_1^\pm, \ldots, x_{D-1}^\pm]^\ell$ is surjective.
Let us show that 
\begin{align}
v_j'' = v_j + v_j' = v_j + A \lambda A^\dagger \lambda v_j
\end{align} 
qualifies by checking
\begin{align}
A^\dagger \lambda v_j'' &= 0 \quad \text{so} \quad  v_j'' \in \im B_0, \label{eq:correctWithA}\\
\text{ and } \quad B_0^\dagger \lambda v_j' &= 0. \label{eq:correctWithB}
\end{align}
\cref{eq:correctWithB} is clear by \cref{lem:bdalg}.
To show \cref{eq:correctWithA}
we note that $Q^{-1} = -\lambda Q^\dagger \lambda$ is also symplectic,
which implies that $Q \lambda Q^\dagger = \lambda$ and $B \lambda A^\dagger = 0$.
Then,
\begin{align*}
v_j' 
&= (A + z B) \lambda A^\dagger \lambda v_j = Q \lambda A^\dagger \lambda v_j,\\
A^\dagger \lambda v_j' 
&= (A^\dagger + \bar z B^\dagger) \lambda v_j'  = Q^\dagger \lambda v_j'\\
&= Q^\dagger \lambda Q \lambda A^\dagger \lambda v_j = \lambda \lambda A^\dagger \lambda v_j\\
&= - A^\dagger \lambda v_j.
\end{align*}
To show that $\Xi$ is injective, suppose $B_0^\dagger \lambda B_0 v = 0$.
Then, $B_0 v \in \im B_0 \cap \ker (B_0^\dagger \lambda) = 
\ker (A_0^\dagger \lambda) \cap \ker (B_0^\dagger \lambda)$ by \cref{lem:bdalg}.
This means that $B_0 v \in \ker (Q^\dagger \lambda) = 0$.
But $B_0$ is injective by construction so $v = 0$.
Now the proof is concluded by the next lemma. 
\end{proof}
\begin{lemma}\label{lem:detInvertibleAntihermitian}
Let $\Xi$ be an invertible antihermitian matrix with zero constant terms in the diagonal 
over $\FF[x_1^\pm,\ldots,x_D^\pm]$ where $\FF$ is any field.
Then, $\det \Xi = c^2$ for some $c \in \FF^\times$.
\end{lemma}
\begin{proof}
Setting all the variables to~$1$ in $\Xi$ we have a skew-symmetric matrix $W = \Xi|_ {x_1 =\cdots = x_D = 1}$
with zero diagonal.
The invertibility implies that $W$ is also invertible, 
and therefore $\det W$ is a square of nonzero Pfaffian, that belongs to $\FF$.
Furthermore, $W$ must be a $2n$-by-$2n$ matrix for some integer $n$, and so is $\Xi$.
Since the determinant is multilinear in the rows, we have $\det(-\Xi^\dagger) = (-1)^{2n} \det(\Xi^\dagger)$.
Since the determinant of any matrix is equal to that of its transpose,
we have $\det(\overline{\Xi}^T) = \det(\overline{\Xi})$.
Since the antipode map is an automorphism of the base ring,
we have $\det(\overline{\Xi}) = \overline{\det \Xi}$.
Since $\Xi = - \Xi^\dagger$ by assumption, 
it follows that $\det \Xi = (-1)^{2n} \det (\Xi^\dagger) = \overline{\det \Xi}$.
This means that $\det \Xi$ is invariant under the antipode map,
but any such element must belong to the coefficient field $\FF$.
Therefore, $\det \Xi = \det W$, which is a nonzero square.
\end{proof}

Given a free module of Pauli operators modulo phase factors,
we can pick a basis (nonredundant generating set)
and make a table of commutation relations among them.
The table is nothing but the matrix $\Xi$ of the preceding lemma.
Note that given a symplectic matrix $Q = A + z B$,
the matrix $\Xi$ is only determined up to congruence;
different free bases $B_0$ and $B_0'$ for $\im B$ are related by an invertible matrix $E$ such that $B_0' = B_0 E$.
Hence, $\Xi' = (B_0')^\dagger \lambda B_0' = E^\dagger (B_0^\dagger \lambda B_0) E = E^\dagger \Xi E$.
See \cite[\S IV.A]{nta3}.
\begin{corollary}\label{cor:bdalgComm}
For any TI Clifford QCA $\alpha$ on $\lat q D$ of prime $p$-dimensional qudits,
the set of all Pauli operators in the boundary algebra of $\alpha$ modulo phase factors
form a free module of even rank.
The antihermitian form of commutation relations on this free module 
has determinant that is a nonzero square of some element of $\FF_p$.
\end{corollary}

Now we discuss some aspects
of our tailored definition of the boundary algebra.
Taking a smaller translation group along directions other than $z$-axis
does alter $\Xi$; even the size of $\Xi$ is not preserved.
The best we could hope for 
is that $\Xi$'s under these choices of smaller group algebras
represent the same element of the Witt group of antihermitian forms that we discuss in the next subsection.
It remains open whether this is the case.
Also, one may wish to consider another axis, say, $y$,
and ask whether the antihermitian form of the boundary algebra at positive $y$-axis
is (stably) congruent to one at the positive or negative $z$-axis.
We do not know the answer to this question either.

For these considerations, 
we do \emph{not} speak of \emph{the} antihermitian form of a TI Clifford QCA $\alpha$.
For a given $\alpha$ and a direction~$z$,
we simply take a small enough translation group,
compose~$\alpha$ with a shift QCA such that the symplectic matrix $Q$ is of form~$A + zB$,
remove redundant columns of~$B$ by \cref{lem:bdalg,lem:Quillen-Suslin-Swan},
and produce an antihermitian matrix.
Nonetheless, given a $z$-direction,
we will prove shortly that taking a smaller translation group along directions other than the $z$-axis
is the only potential ambiguity.

\subsection{Circuit invariance of induced antihermitian forms}

\begin{lemma}\label{lem:CircuitsGiveTrivialAntihermitianForms}
Let $\alpha$ be a TI Clifford QCA, and $\beta$ be a TI Clifford circuit composed with a TI shift QCA.
Suppose $\alpha$ has an associated antihermitian form $\Xi$ at positive $z$-axis.
Then, there exist a smaller translation group generated by $x_1,\ldots,x_{D-1},z^\ell$
($\ell$ is the new unit cell size under the smaller translation group)
and TI shift QCA $\gamma$ such that
$\alpha \beta \gamma$ has an associated antihermitian form at the positive $z$-axis
that is congruent to $\diag(\Xi, \lambda_n)$ for some $n \ge 0$. 

In particular, if $\alpha$ is a TI Clifford circuit composed with shifts, 
then under some choice of smaller translation group
an antihermitian form of $\alpha$ is congruent to $\lambda_n$ for some $n \ge 0$.
\end{lemma}
\begin{proof}
Suppose $Q = A + zB$ is a symplectic matrix of $\alpha$ over $\FF[x_1^\pm,\ldots, x_{D-1}^\pm, z^\pm]$
where $A,B$ do not involve $z$.
We have shown that $B^\dagger \lambda B \cong \diag( \Xi, 0 )$ where $\det \Xi$ is a nonzero square of $\FF$.

First, we claim that 
taking a smaller $z$-translation group does not change the congruence class of $\Xi$.
Indeed, for example, if we let two sites along $z$-axis
of the original lattice to form one supersite in a coarser lattice,
then the corresponding Laurent polynomial description is obtained by replacing every $z$ with 
the matrix $\phi_2(z) = \begin{pmatrix} 0 & z' \\ 1 & 0 \end{pmatrix}$,
which is the matrix representation of $z$ as an automorphism of $\FF[z^\pm]$ 
that is viewed as a rank~2 free module over $\FF[z^{\pm 2}]$.
That is, $Q = A + z B$ over the coarser translation group algebra becomes 
$A \otimes \begin{pmatrix} 1 & 0 \\ 0 & 1 \end{pmatrix}  + B \otimes \begin{pmatrix} 0 & z' \\ 1 & 0 \end{pmatrix}$.
The new antihermitian matrix with respect to this smaller translation group
is computed from $B \otimes \begin{pmatrix} 0 & 1 \\ 0 & 0 \end{pmatrix}$
and thus is unchanged up to congruence.
More generally, any representation $\phi_n(z)$ of $z$ as an automorphism of a smaller group algebra
is a similarity transformation of
\begin{align}
\begin{pmatrix}
0 & 0 & \cdots & 0 & z' \\
1 & 0 & \cdots & 0 & 0 \\
0 & 1 & \ddots & 0 & 0 \\
0 & 0 & \hdots & 1 & 0
\end{pmatrix}
\end{align}
and the antihermitian form will be a congruence transform of the original $\Xi$.

Next, if $\beta$'s symplectic matrix $E$ does not involve $z$,
then $\alpha \beta$ has the symplectic matrix $QE = AE + z BE$ and 
$E^\dagger B^\dagger \lambda B E \cong \diag( \Xi, 0 )$;
we don't even need the assumption that $E$ is elementary.
Hence, the only nontrivial case is when $E$ involves $z$.

Observe that Control-Phase and Control-X in \cref{eq:ElemSymp} 
is a product of $S E' S^{-1}$ under suitably smaller translation group,
where $E'$ does not involve any variables, and $S$ is a shift.
Indeed, under a smaller translation group Control-Phase is either a single qubit operator 
(repeated translation invariantly),
Control-X conjugated by Hadamard,
or a product thereof.
Control-X is of form
\begin{align}
\begin{pmatrix}
1 & 0 &   &    \\
m & 1 &   &    \\
  &    & 1 & - m^{-1} \\
  &   & 0 & 1
\end{pmatrix}
=
\begin{pmatrix}
m^{-1} &  &   &    \\
 & 1 &   &    \\
 &     & m^{-1} &  \\
  &    &  & 1
\end{pmatrix}
\begin{pmatrix}
1 & 0 &   &    \\
1 & 1 &   &    \\
   &   & 1 & -1 \\
    &  & 0 & 1
\end{pmatrix}
\begin{pmatrix}
m &  &   &    \\
 & 1 &   &    \\
 &     & m &  \\
 &     &  & 1
\end{pmatrix} \label{eq:cnotDecomp}
\end{align}
where $m$ is any monomial, or a product thereof.
Hadamard does not involve any variable.
Hence, any elementary symplectic matrix is a product of shifts (the extra gate $J$) 
and elementary symplectic matrices that lack any variable.

Therefore, it suffices for us to prove the lemma when $\beta$ is a shift along $z$-direction,
i.e., $E$ is the extra gate $J$.
Since taking smaller translation group along $z$-direction does not change the antihermitian form,
we only have to consider $E = E_0 + z E_1$ where $E_0,E_1$ are diagonal matrices with $0,1$ entries;
see e.g. the last matrix of \cref{eq:cnotDecomp}.
Note that $E_0$ and $E_1$ are mutually orthogonal projectors 
and that there is some $n \ge 0$ such that
\begin{align}
E_1^\dagger \lambda E_1 \cong \lambda_n \oplus 0. \label{eq:E1lambda}
\end{align}
Now let us expand $QE$ to compute the antihermitian form.
\begin{align}
&(A+zB)(E_0 + z E_1) = AE_0 + z BE_0 + z AE_1 + z^2 B E_1 \nonumber \\
\xrightarrow{z^2 \to z'}&
A E_0 \otimes \begin{pmatrix} 1 & 0 \\ 0 & 1 \end{pmatrix} 
+ (B E_0 + AE_1) \otimes \begin{pmatrix} 0 & z' \\ 1 & 0 \end{pmatrix} + 
B E_1 \otimes \begin{pmatrix} z' & 0 \\ 0 & z' \end{pmatrix}.
\end{align}
So, the new antihermitian form is the ``nonsingular part'' that is
provided by \cref{lem:bdalg,lem:Quillen-Suslin-Swan} from
\begin{align}
&\begin{pmatrix}
E_1^\dagger B^\dagger & 0 \\ E_0^\dagger B^\dagger + E_1^\dagger A^\dagger & E_1^\dagger B^\dagger
\end{pmatrix}
\begin{pmatrix}
\lambda & \\ & \lambda
\end{pmatrix}
\begin{pmatrix}
B E_1 & B E_0 + A E_1 \\ 0 & B E_1
\end{pmatrix}
\end{align}
This simplifies using $A^\dagger \lambda A + B^\dagger \lambda B = \lambda$,
which follows from 
$Q^\dagger \lambda Q = \lambda$,
to give
\begin{align}
=&\begin{pmatrix}
E_1^\dagger B^\dagger \lambda B E_1 & E_1^\dagger B^\dagger \lambda B E_0 \\
E_0^\dagger B^\dagger \lambda B E_1 & 
E_0^\dagger B^\dagger \lambda B E_0 
+ E_1^\dagger A^\dagger \lambda AE_1 
+ E_1^\dagger B^\dagger \lambda B E_1
\end{pmatrix} \\
=&
\begin{pmatrix}
E_1^\dagger B^\dagger \lambda B E_1 & E_1^\dagger B^\dagger \lambda B E_0 \\
E_0^\dagger B^\dagger \lambda B E_1 & 
E_0^\dagger B^\dagger \lambda B E_0 
+ E_1^\dagger \lambda E_1 
\end{pmatrix}  \nonumber
\end{align}
which is rewritten as
\begin{align}
=&
\begin{pmatrix}
E_1^\dagger & \\ & E_0^\dagger
\end{pmatrix}
\begin{pmatrix}
B^\dagger\lambda B & B^\dagger\lambda B\\ B^\dagger\lambda B & B^\dagger\lambda B
\end{pmatrix}
\begin{pmatrix}
E_1 & \\ & E_0
\end{pmatrix}
+
\begin{pmatrix}
0 & \\ & E_1^\dagger \lambda E_1
\end{pmatrix}.
\end{align}
Here, the first and second summands act on disjoint sets of components due to projectors $E_0$ and $E_1$.
The second summand is $\lambda_n \oplus 0$ as noted in \cref{eq:E1lambda}.
The first summand is a row and column rearrangement of $B^\dagger \lambda B$
whose ``nonsingular'' part is the antihermitian form of $Q$.
\end{proof}

We note a fact that sometimes facilitates calculation.
\begin{lemma}
Let $Q = A+ z B$ be a $2q \times 2q$ symplectic matrix where $A,B$ does not involve variable $z$.
Let $B_0$ be the last $q$ columns of $B$, and put $\Xi = B_0^\dagger \lambda_q B_0$.
If the matrix $\Xi$ is invertible, 
then $\Xi$ is an antihermitian form of \cref{cor:bdalgComm}.
\end{lemma}
The last $q$ columns of a symplectic matrix $Q$ corresponds to a (locally flippable) separator~\cite{nta3}.
This lemma allows to determine an antihermitian form 
without explicitly knowing the first $q$ columns of~$Q$,
beyond their existence.
\begin{proof}
Since $\Xi \Xi^{-1} = I$, we see $B_0^\dagger \lambda (B_0\Xi^{-1}) = I$.
%Here, $B_0\Xi^{-1}$ is a matrix that does not involve the variable $z$.
If $A_0$ is the last $q$ columns of $A$ 
(so $Q_0 = A_0 + z B_0$ is the last $q$ columns of $Q$),
we have $(A_0 + z B_0)^\dagger \lambda (-z B_0 \Xi^{-1}) = -I$
since $A^\dagger \lambda B = 0$. Put $\Xi^{-\dagger} = E^\dagger - E$
so that
\begin{align}
\left(\begin{array}{c|c}
-z B_0 \Xi^{-1} + Q_0 E & Q_0 
\end{array}\right)^\dagger \lambda
\left(\begin{array}{c|c}
-z B_0 \Xi^{-1} + Q_0 E & Q_0
\end{array}\right)
= 
\left(\begin{array}{cc}
0 & I \\ -I & 0
\end{array}\right) = \lambda.
\end{align}
Note that neither $\Xi^{-1}$ nor $E$ involves the variable $z$.
The constructed symplectic matrix is thus in the form 
\begin{align}
\left(\begin{array}{c|c}
-z B_0 \Xi^{-1} + Q_0 E & Q_0
\end{array}\right) =
\left(\begin{array}{c|c}
 A_0 E & A_0
\end{array}\right) +
z\left(\begin{array}{c|c}
-B_0 \Xi^{-1} + B_0 E & B_0
\end{array}\right),
\end{align}
where the column span of the coefficient matrix of $z$ is the same as that of $B_0$.
Therefore, $\Xi$ is an antihermitian form induced from this new symplectic matrix.

But all symplectic matrices with the same right (or left) half 
are equivalent up to a product of elementary symplectic matrices~\cite[Lem.IV.10]{nta3}.
\cref{lem:CircuitsGiveTrivialAntihermitianForms} now guarantees that $\Xi$ is what we desire.
\end{proof}

\subsection{Witt group of antihermitian forms}

In this subsection we digress from TI Clifford QCA, but focus on antihermitian forms over $R$.

Recall that the Witt group of a field (of characteristic different from~$2$) 
is the group of all equivalence classes of nondegenerate symmetric bilinear forms~\cite{Lam,Kniga}.
The equivalence relation is established by treating hyperbolic planes as trivial elements.
We can define an analogous group for antihermitian forms over $R = \FF[x_1^\pm,\ldots,x_D^\pm]$.
We declare that invertible antihermitian matrices $\Xi$ and $\Xi'$ of possibly different dimensions
are {\bf equivalent} and write $\Xi \simeq \Xi'$
if there is an invertible matrix $E$ and nonnegative integers $n,n'$ such that
\begin{align}
E^\dagger (\Xi \oplus \lambda_n) E = \Xi' \oplus \lambda_{n'}. \label{eq:stableEquivalenceOfXi}
\end{align}
Then we consider the abelian monoid of all equivalence classes of 
finite dimensional invertible antihermitian matrices (with zero constant terms in the diagonal) over $R$
where the monoid operation is given by the orthogonal (direct) sum.
This monoid is actually a group, that we call {\bf Witt group of antihermitian matrices},
as the following lemma constructs inverses regardless of whether $\FF$ has characteristic~2.
\begin{lemma}\label{lem:inverseXiExists}
Write $\Xi \cong \Xi'$ if $E^\dagger \Xi E = \Xi'$ for some invertible $E$ over $R$.
Let $\Xi = - \Xi^\dagger$ be an invertible antihermitian matrix with no constant term in the diagonal.
Then,
\begin{align}
\begin{pmatrix} \Xi & 0 \\ 0 & -\Xi  \end{pmatrix}
\cong
\begin{pmatrix} \Xi & 0 \\ 0 & \Xi^{-1}  \end{pmatrix} 
\cong 
\lambda.
\end{align}
\end{lemma}
The no-diagonal-constant-term condition is automatically met whenever $2 = 2^\dagger$ is a unit.
\begin{proof}
The first congruence is by $\Xi$ on the second direct summand.
Since $\Xi$ has no constant diagonal term,
we may find $M$ such that $\Xi = M - M^\dagger$.
Then, $E = \begin{pmatrix}I & \Xi^{-1} \\ 0 & I\end{pmatrix}
\begin{pmatrix} I & 0 \\ -M & I\end{pmatrix}$ is a desired matrix for the second congruence.
\end{proof}

The exponent of TI Clifford QCA is consistent 
with the exponent of the Witt group of antihermitian forms.
\begin{lemma}\label{lem:WittExponentBound}
Let $\FF$ be any field of characteristic not $2$,
and $R \supseteq \FF$ be a ring with an involution that fixes $\FF$ elementwise.
If the Witt group of $\FF$ has exponent $n$,
then the Witt group of invertible antihermitian matrices over $R$ has exponent $n$.
\end{lemma}
\begin{proof}
Let $\Xi = - \Xi^\dagger$ be an invertible antihermitian matrix over $R$.
The $n$-th power of $\Xi$ in the Witt group is represented by $I_n \otimes \Xi$
where $\otimes$ is the Kronecker product.
The assumption is that $I_n$ is congruent over $\FF$ to a direct sum of 
$\begin{pmatrix}
0 & 1 \\ 1 & 0
\end{pmatrix}$. Necessarily, $n$ is an even number.
This implies that
\begin{align}
I_n \otimes \Xi
\cong 
\begin{pmatrix}
0 & 1 \\ 1 & 0
\end{pmatrix}^{\oplus (n/2)}
\otimes \Xi.
\end{align}
By observing
\begin{align}
\begin{pmatrix}
I & 0 \\ 0 & \Xi^{- \dagger} 
\end{pmatrix}
\begin{pmatrix}
0 & \Xi \\ -\Xi^\dagger & 0
\end{pmatrix}
\begin{pmatrix}
I & 0 \\ 0& \Xi^{-1}
\end{pmatrix}
=
\begin{pmatrix}
0 & I \\ -I & 0
\end{pmatrix}
\end{align}
the proof is complete.
\end{proof}

\begin{lemma}\label{lem:WittGroupComputationForD01}
The Witt group of antihermitian matrices over any field $\FF$ ($D=0$) is trivial.
The Witt group of antihermitian matrices over $\FF[x^\pm]$ ($D=1$) is trivial if $\FF$ is a finite field.
\end{lemma}
\begin{proof}
The case of $D=0$ is well known; see e.g. \cite[XV.8.1]{Lang}.

Since $\FF[x^\pm]$ is a principal ideal domain, 
it suffices to show that there is a column matrix $v \neq 0$ such that $v^\dagger \Xi v = 0$ (isotropy).
The reason is that if such $v$ exists then we can factor out 
the common divisor of the the entries of $v$
to obtain another isotropic column matrix $v'$,
and by Euclid's algorithm we can find an invertible matrix $E$ whose first column is $v'$.
Then, $E^\dagger \Xi E$ has the top left corner zero, 
and further application of Euclid's algorithm 
focusing on the first row of $E^\dagger \Xi E$ should reveal~$1$ in the $(1,2)$ entry,
yielding a hyperbolic plane, allowing induction in the dimension of $\Xi$.

The isotropy can be tested by taking the quotient field $\FF(x)$:
if $v^\dagger \Xi v = 0$ for some nonzero $v$ with entries in $\FF(x)$,
then $(cv)^\dagger \Xi (cv) = 0$ where $c$ is a common multiple of all denominators of the entries of $v$,
and $cv$ belongs to $\FF[x]$ and is nonzero.
Therefore, it suffices to consider $\Xi$ over $\FF(x)$.

It is not allowed that $\rank \Xi = 1$ 
because the invertibility forces the sole entry be a unit of $\FF[x^\pm]$,
and in turn the antihermicity forces it be in $\FF$ (actually zero if $1 \neq -1 \in \FF$),
but we are assuming that the constant term in the diagonal be zero.

If $\rank \Xi = 2$ and if $\Xi$ is anisotropic, 
then we can diagonalize it as $\Xi \cong \diag(h,1/h)$
with $h = - h^\dagger$, but $v = (1,h)$ satisfies $v^\dagger \diag(h,1/h) v = 0$, which is absurd.
So, any rank~$2$ invertible antihermtian matrix is isotropic.

In general, $\Xi$ is isotropic if and only if $(x-\bar x) \Xi$ is isotropic,
where the latter is a hermitian matrix of nonzero determinant.
Let $\FF(y)$ be the involution-invariant subfield of~$\FF(x)$.
{\it i.e.}, $\FF(y) = \FF(x + \bar x)$  where~$y = x + \bar x$,
so $\FF(x) = \FF(y)[x]/(x^2 - y x + 1)$ is a degree~$2$ field extension over~$\FF(y)$.
Given a hermitian form~$\Phi$ of dimension~$m$,
we can consider
a quadratic function~$\xi : \FF(x)^m = \FF(y)^m \oplus x \FF(y)^m = a + x b \mapsto (a^\dag + \bar x b^\dag) \Phi (a + x b)$
which is valued in the involution-invariant subfield~$\FF(y)$.
So, we have a quadratic form~$\xi$ over~$\FF(y)$ of dimension~$2m$.
Tautologically $\Phi$ is anisotropic if and only if $\xi$ is anisotropic.

For a finite field $\FF$ of odd characteristic,
any nondegenerate symmetric form over $\FF(y)$ of dimension~$>4$ is isotropic~\cite[XI.1.5]{Lam}.
This implies that $\Xi$ of dimension $> 2$ is isotropic.
For $\FF$ of characteristic~2,
one can ``diagonalize''~$\xi$ so that $\xi$ is an orthogonal sum 
of two- and one-dimensional quadratic forms~\cite[7.31]{Kniga}.
--- Think of $\xi$ as a $2m$-by-$2m$ matrix over~$\FF(y)$ up to equivalence $\xi \sim \xi + \phi - \phi^T$.
Any off diagonal element in the matrix gives a two-dimensional subspace that is orthogonal to all the rest 
after a congruence transformation.
Hence, by induction, one is left with two- and one-dimensional subforms. ---
The anisotropy implies that, in particular, 
a set of ``diagonal elements,'' one diagonal element from each two-dimensional block
and one diagonal element from each one-dimensional block, 
must be linearly independent over~$\FF(y^2)$~\cite[Lem.\,36.8]{Kniga}.
Such a set may contain $\tfrac 1 2 (2m)$ elements or more, 
which must be at most the field extension degree~$[\FF(y):\FF(y^2)] = 2$.
Hence, the anisotropy of~$\xi$ implies that~$m \le 2$.
\end{proof}

Note that in usual Witt groups over fields 
one declares that a symmetric form $\Xi$ is equivalent to another $\Xi'$
if either $\Xi$ is congruent to $\Xi' \oplus H'$ or $\Xi'$ is congruent to $\Xi \oplus H$ for some hyperbolic spaces $H, H'$.
We only demand stable equivalence in \cref{eq:stableEquivalenceOfXi}.
We do not know if our equivalence implies the nonstabilized version.
This is because we do not have an analog of the Witt cancellation theorem.

\subsection{Boundary algebra determines QCA}

We now show that the boundary algebra 
contains all the information about the ``bulk'' TI Clifford QCA
modulo Clifford circuits and shifts.
To this end, we first establish the equivalence 
between two notions of equivalence of QCA.
The first notion is what we have defined using (TI Clifford) quantum circuits and shifts.
The second is what follows.

Let us say two QCA $\alpha$ and $\beta$ of finite ranges on a system {\bf blend}~\cite{FreedmanHastings2019QCA},
or $\alpha$ {\bf blends into} $\beta$,
from a subsystem $A$ to another $B$
if there exists an {\bf interpolating} QCA $\gamma$ of finite range
such that $\gamma$ agrees with $\alpha$ on $A$ and with $\beta$ on $B$.
We will only consider cases where the regions $A$ and $B$ are both infinite.

\begin{lemma}\label{lem:gappable-triviality}
Suppose a translation invariant Clifford QCA $\alpha$ of range $r$ 
on $\lat q D$ of prime $p$-dimensional qudits
blends into a shift QCA
from a half space $\{ \vec x : x_D \le - r \}$
to another disjoint half space $\{ \vec x :  x_D > 0\}$
by a Clifford QCA $\gamma$ that is weakly translation invariant 
along all directions orthogonal to the $x_D$-axis.
Then, $\alpha$ is a Clifford circuit that is weakly translation invariant, followed by a shift.
\end{lemma}
The converse is true and obvious: 
Given a circuit, we can drop all the gates on a half space 
and the remaining gates define an interpolating QCA.
Ref.~\cite[\S 3.5]{FreedmanHastings2019QCA} shows a similar result for a more general class of QCA.
However, to the present author's understanding, 
the argument there is not applicable to this lemma.
See \cref{app:trivialQCA}.
\begin{proof}
Without loss of generality we assume that the shift QCA $\beta$ into which $\alpha$ blends
is actually the identity QCA;
if $\beta$ is nontrivial, we instead consider $\beta^{-1}\alpha$ in place of $\alpha$
where $\beta^{-1}\alpha$ blends into the identity QCA.

Shifting $\gamma^{-1}$ along $x_D$-axis by distance $10r$ to the negative $x_D$-direction,
and composing the shifted $\gamma^{-1}$ with the unshifted $\gamma$,
we obtain a QCA $\alpha_P$ that interpolates $\Id$ on $\{ \vec x : x_D \le -11r \}$,
$\alpha$ on the quasi-hyperplane $P = \{ \vec x : -10 r \le x_D \le -r \}$,
and $\Id$ on $\{ \vec x : x_D \ge 0\}$,
since $\alpha$ is translation invariant.
Similarly by shifting $\alpha_P$ we have $\alpha_{P_j}$ for any $j \in \ZZ$
that interpolates $\Id$, $\alpha$, and $\Id$,
which is supported on
\begin{align}
P^+_j = \{ \vec x : -11 r + 20r j \le x_D \le 20r j \}.
\end{align}

Since each $\alpha_{P_j}$ is supported on the quasi-hyperplane $P^+_j$,
we may regard $\alpha_P$ as a TI Clifford QCA on $\lat{12rq}{D-1}$.
But we have proved in \cref{sec:exponents} that the tensor product of four 
TI Clifford QCA on a system of prime dimensional qudits
is always a Clifford circuit followed by a shift.
Since the quasi-hyperplanes $P^+_j$ are disjoint,
this means that the product $\alpha_{P_{j}}\alpha_{P_{j+1}}\alpha_{P_{j+2}}\alpha_{P_{j+3}}$
is a Clifford circuit followed by a shift.
Thus, we trivialize the product of all $\alpha_{P_j}$.

Composing this trivializing circuit plus shift with $\alpha$,
we are left with a QCA that is a product of disjoint QCA,
supported on the $2r$-neighborhood of the complement of all $P^+_j$.
The argument of the preceding paragraph trivializes this remainder.
\end{proof}

The blending is happening at the geometric boundary of the lattice,
and hence the possibility of blending between two QCA 
should be regarded as a property of the boundary algebras.
Let us make this intuition more precise.

\begin{lemma}\label{lem:trivialXi-gives-trivialQCA}
Let $\Xi$ be an invertible antihermitian form induced from a TI Clifford QCA $\alpha$.
If $\Xi \cong \lambda$, then $\alpha$ is a weakly TI Clifford circuit followed by a shift.
\end{lemma}
\begin{proof}
Without loss of generality, we may assume that the symplectic matrix of $\alpha$
is $Q = A + z B$ where $A,B$ do not involve $z$.
We will construct a blending QCA $\gamma$ that interpolates $\alpha$ on the space with $z \ge 0$ 
and some shift QCA on the space with sufficiently negative $z$;
the blending region is some slab of finite thickness.
Then, the claim will follow from \cref{lem:gappable-triviality}.

Let $\HH = \{ (x_1,x_2,\ldots,x_{D-1},z) \in \ZZ^D : z \ge 0 \}$ be a half space.
Let $\Mat \HH$ be the algebra of all operators with finite support contained in $\HH$.
The image $\alpha(\Mat \HH)$ is also supported on $\HH$ since $Q$ has only nonnegative exponents on $z$.
Since $Q^{-1} = -\lambda Q^\dagger \lambda = -\lambda A^\dagger \lambda - z^{-1} \lambda B^\dagger \lambda$,
any operator whose support has strictly positive $z$ coordinates maps under $\alpha^{-1}$ into $\Mat \HH$;
such an operator cannot be in the commutant of $\alpha(\Mat \HH)$ within $\Mat \HH$.
Hence, the commutant $\calC$ of $\alpha(\Mat \HH)$ within $\Mat \HH$
is supported on the $z=0$ plane, and is equal to the commutant within $z=0$ plane 
of the operators corresponding to columns of $A$.
That is, the commutant $\calC$ corresponds to $\ker (A^\dagger \lambda)$,
which is $\im B$ by \cref{lem:bdalg},
which also says that the smallest nonzero determinantal ideal of $B$ is unit.
Then, \cref{lem:Quillen-Suslin-Swan} implies that $\im B$ has a free basis
which can be written in the columns of a matrix $B_0$,
i.e., $B_0$ has full rank and $\im B = \im B_0$.
The boundary antihermitian form is then $\Xi = B_0^\dagger \lambda_q B_0$.
The assumption $\Xi \cong \lambda$ implies that we may assume that $\Xi = \lambda_n$ for some $n$.
This means that the generating set (over the translation group algebra) of $\im B$ is partitioned
into $n$ pairs $(\tilde X_s,\tilde Z_s)$ , $s = 1,\ldots,n$ 
where each pair obeys the commutation relation of $X,Z$ of a single $p$-dimensional qudit.

We define a blending $\gamma$ as follows.
For $P\in \Mat \HH$, we let $\gamma(P) = \alpha(P)$.
Due to the translation invariance in $x_1,\ldots,x_{D-1}$ directions,
we only have to define $\gamma$ for a fixed $\vec x = (x_1,\ldots,x_{D-1})$.
On the semi-infinite line $\ell(\vec x) = \{(\vec x, z): z < 0 \}$,
we choose a set $S(\vec x) \subset \ell(\vec x)$ of $n$ qudits
on which there are single-qudit Pauli operators $X_s,Z_s$.
We define $\gamma(X_s) = \tilde X_s$ and $\gamma(Z_s) = \tilde Z_s$ for $s = 1,\ldots, n$.
Though the particular choice of $S(\vec x)$ does not matter 
(but we insist on the translation invariance with respect to $\vec x$), 
one may choose those that are as close to $\HH$ as possible.
The qudits in $\ell(\vec x) \setminus S(\vec x)$ are ordered 
(first by $z$ coordinate and then by some fixed ordering within a unit cell),
so the operators on $\ell(\vec x) \setminus S(\vec x)$ are, as part of definition of $\gamma$,
shifted to the operators on $\ell(\vec x)$.
Note that any operator on $\ell(\vec x)$ has the obvious inverse image in this definition of $\gamma$.

Thus constructed $\gamma$ is defined everywhere and agrees with $\alpha$ on $\HH$
and acts as a shift QCA on the space with $z$ coordinates $< -n$.
It is obvious that $\gamma$ is a $*$-monomorphism of the algebra of finitely supported operators and is Clifford.
To show that $\gamma$ is surjective,
we find a preimage for each single-qudit Pauli operator $P$ at $(\vec x, z)$.
If $P$ has $z > 0$, then it has a preimage under $\gamma$
because it does under $\alpha$.
If $P$ has $z < 0$, then $P$ is a shift of some operator with $z < 0$.
Suppose $P$ is on $z=0$ plane, and let $p$ be the Laurent polynomial column matrix of $P$.
From now on until the end of this proof, 
we do not distinguish an operator from the corresponding Laurent polynomial column matrix.
Define $p|_A = -A \lambda A^\dagger \lambda p \in \im A$ and $p|_B = - B \lambda B^\dagger \lambda p \in \im B$.
Expanding $I = Q Q^{-1} = Q (-\lambda A^\dagger \lambda - z^{-1} \lambda B^\dagger \lambda)$
in $z$ we see that $I = -A\lambda A^\dagger \lambda - B \lambda B^\dagger \lambda$.
(We have noted this identity in the proof of \cref{lem:detXi}.)
Hence, $p|_A + p|_B = p$.
Now, $p|_B \in \im B$ certainly has a preimage under $\gamma$ on $\bigcup_{\vec x} S(\vec x)$.
In addition, $p|_A = QQ^{-1} p|_A = Q ( -\lambda A^\dagger \lambda p|_A + z^{-1} \lambda B^\dagger \lambda p|_A) = Q ( -\lambda A^\dagger \lambda p|_A)$
because $B^\dagger \lambda a = 0$ for any $a \in \im A$.
Here, $-\lambda A^\dagger \lambda p|_A$ is a preimage of $p|_A$ under $\gamma$
since it does not involve any $z$,
i.e., the operator is on $z=0$ plane.
Hence, both $p|_A$ and $p|_B$ have preimages under $\gamma$, so does $p$.
\end{proof}

\begin{lemma}\label{lem:inverse-qca-gives-inverse-Xi}
Let $\Xi$ be an invertible antihermitian form 
induced from the symplectic matrix $A + z B$ of a TI Clifford QCA $\alpha$,
and $\Xi'$ be that from $Q^{-1} = -\lambda(A^\dagger + z^{-1} B^\dagger)\lambda$ of $\alpha^{-1}$, after a shift.
Then, $\Xi^{-1} \simeq \Xi'$.
\end{lemma}
\begin{proof}
We know $\alpha \otimes \alpha^{-1}$ is a weakly TI Clifford QCA;
this follows from the same argument~\cite{Arrighi2007} that 
$\alpha \otimes \alpha^{-1} = (\alpha \otimes \Id) (\SWAP) (\alpha^{-1} \otimes \Id) (\SWAP)$
is a circuit, with the observation that $(\alpha \otimes \Id) (\SWAP_{s,s'}) (\alpha^{-1} \otimes \Id)$
where $\SWAP_{s,s'}$ swaps two sites is a Clifford circuit.
Therefore,
$Q \oplus Q^{-1}$ is a product of elementary symplectic transformations that are weakly TI.
By \cref{lem:CircuitsGiveTrivialAntihermitianForms}
we know $\Xi' \oplus \Xi$ has to be trivial.
\end{proof}

The next lemma shows that every invertible antihermitian form determines a valid boundary algebra.

\begin{lemma}\label{lem:XitoQCA}
For any invertible antihermitian form $\Xi$ over $\FF[x_1^\pm,\ldots,x_{D-1}^\pm]$
there exists a $D$-dimensional TI Clifford QCA $\alpha$ 
that induces $\Xi$ at the positive $x_D$-axis.
\end{lemma}
\begin{proof}
Let $q$ be the dimension of $\Xi$.
\Cref{lem:inverseXiExists} gives an invertible matrix $E = \begin{pmatrix} A & B \end{pmatrix}$
over $\FF[x_1^\pm,\ldots,x_{D-1}^\pm]$
where $A$ and $B$ are $2q$-by-$q$ submatrices of $E$ consising of first and last $q$ columns, respectively,
such that 
\begin{align}
\begin{pmatrix}-\Xi & 0 \\ 0 & \Xi\end{pmatrix}
= 
\begin{pmatrix} A^\dagger \\ B^\dagger \end{pmatrix}
\begin{pmatrix}0 & I_q \\ - I_q & 0 \end{pmatrix}
\begin{pmatrix} A & B \end{pmatrix}.
\end{align}
Observe that the matrix $G = \begin{pmatrix} A & x_D B \end{pmatrix}$ also satisfies
the same equation with $G$ in place of $E$.
It follows that $G E^{-1}$ is symplectic and defines a desired $D$-dimensional QCA.
\end{proof}

\begin{corollary}\label{cor:Xi-to-QCA}
The stable equivalence class (as defined in \cref{eq:stableEquivalenceOfXi})
of any invertible antihermitian matrix $\Xi$ over $\FF_p[x_1^\pm,\ldots,x_{D-1}^\pm]$
uniquely determines a $D$-dimensional TI Clifford QCA $\alpha_\Xi$
up to the equivalence of \cref{def:equivQCA}
which induces $\Xi$ at the positive $x_D$-axis.
\end{corollary}
\begin{proof}
Let $\alpha$ and $\beta$ be two TI Clifford QCA constructed by \cref{lem:XitoQCA} 
from two invertible antihermitian matrices $\Xi_1$ and $\Xi_2$, respectively,
where $\Xi_1$ and $\Xi_2$ are stably equivalent.
By \cref{lem:inverse-qca-gives-inverse-Xi},
$\gamma = \alpha \otimes \beta^{-1}$ induces $\Xi_1 \oplus \Xi_2^{-1}$,
which is assumed to be trivial in the presence of possibly a shift QCA.
(This additional shift QCA is to supply the stabilization in our notion of equivalence;
see \cref{eq:stableEquivalenceOfXi}.)
By \cref{lem:trivialXi-gives-trivialQCA}, $\gamma$ is trivial.
\end{proof}

\subsection{Triviality of Clifford QCA in 2D}\label{sec:trivial2dqca}

\begin{proof}[Proof of \cref{thm:main}~(i)]
Any antihermitian form induced from a 2D TI Clifford QCA is over a one-variable Laurent polynomial ring
and is invertible (\cref{cor:bdalgComm}).
Such an antihermitian form is always congruent to $\lambda_n$ for some $n\ge 0$ (\cref{lem:WittGroupComputationForD01}).
This means that any 2D TI Clifford QCA always blends into a TI shift QCA 
and thus it is trivial (\cref{lem:trivialXi-gives-trivialQCA}).
\end{proof}

Note that a special property of 2D enters this triviality proof
only through \cref{lem:WittGroupComputationForD01},
which is conceptually saying that it is possible to ``gap out'' 1D boundary of a 2D system.
It is highly nontrivial to ``gap out'' the boundary in general.

\begin{remark}\label{rem:necessityOfStabilization}
We do not know whether one can show the triviality of a TI Clifford QCA in 2D
while preserving both the full translation symmetry of a QCA and the number of qudits per unit cell
(i.e., no translation symmetry breaking and no stabilization with ancillas).
If one insists on the full translation symmetry in every layer of Clifford gates,
then one should find a decomposition of the following symplectic matrix $Q$ of a QCA on $\lat 2 2$ 
into elementary matrices associated with Clifford circuit layers.
\begin{align}
Q = \begin{pmatrix} E & 0 \\ 0 & E^{-\dagger} \end{pmatrix} 
\quad \text{where} \quad
E = \begin{pmatrix} 1 + xy & x^2 \\ -y^2 & 1 - xy \end{pmatrix}.
\end{align}
The matrix $E$ here is the Cohn matrix~\cite[\S 7]{Cohn1966}
that has determinant $1$ but is {\em not} a product of elementary row operations.
This means that the TI Clifford QCA associated with $Q$
cannot be written as a circuit of control-Not gates 
where each layer obeys the full translation structure.
We have avoided this issue by weak translation symmetry breaking.
Indeed, $Q$ is a circuit of control-Not gates 
where at least one layer does not obey the full translation group
but only a subgroup of index~2 by Suslin's stability theorem~\cite[Cor.7.11]{Suslin1977Stability}.

In spirit of Suslin's result,
one may ask whether a 2D TI Clifford QCA with a sufficiently large unit cell
is a Clifford circuit with full translation symmetry in every layer in the circuit.
$\square$
\end{remark}

\section{Clifford QCA in 3D}\label{sec:3DQCA}

In this section, we construct examples of nontrivial TI Clifford QCA in three dimensions.
We will present explicit examples in terms of symplectic matrices,
explain the idea behind the construction,
and then prove that they are nontrivial.
For the nontriviality, it is unfortunate that we cannot rely on our theory of boundary algebras.
The hindering piece is the ambiguity of the antihermitian form under weak translation symmetry breaking. (See the discussion at the end of \cref{sec:constructXi}.)
So, we resort to a surface topological order which is manifestly scale invariant.

To define a TI Clifford QCA on $\lat q D$ 
it suffices to specify a TI commuting Pauli Hamiltonian 
(Pauli stabilizer Hamiltonian) on $\lat q D$
that has unique ground state on a torus of sufficiently large size~\cite[Thm.IV.4]{nta3};
the Hamiltonian terms can be rearranged 
so that they form a locally flippable separator,
and a QCA can be constructed and is unique up to a TI Clifford circuit~\cite[\S IV.B]{nta3}.
The most compact representation of such a Hamiltonian is by the polynomial method~\cite{Haah2013}
which we used in \cref{sec:exponents}.
For example,
the 2D toric code is represented by a stabilizer map
\begin{align}
\sigma_{toric} = 
\left(
\begin{array}{c|c}
 y-1 & 0 \\
 1-x & 0 \\
 \hline
 0 & \frac{1}{x}-1 \\
 0 & \frac{1}{y}-1 \\
\end{array}
\right). \label{eq:toricCode}
\end{align}
Each monomial is a Pauli matrix where the coefficient determines whether it is $X$ or $Z$
and the exponents of variables denote the location.
Compare this equation with \cref{fig:toricCode}.
\begin{figure}[hb]
\includegraphics[width=0.6\textwidth, trim ={0cm 16cm 26cm 0cm},clip]{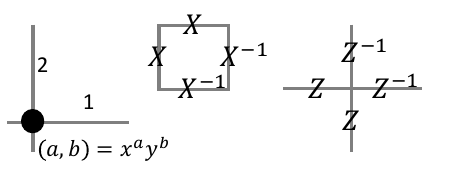}
\caption{
Toric code terms. This is a transcription of \cref{eq:toricCode}.
A similar diagram can be drawn for \cref{eq:expandedSigmaQCA}.
}
\label{fig:toricCode}
\end{figure}

Our Hamiltonian is defined on $\lat 2 3$ of {\em odd} prime $p$-dimensional qudits
and there are exactly two types of terms per site.
The stabilizer map $\sigma_{QCA}$ is given by a $4$-by-$2$ matrix over $R = \FF_p[x^\pm,y^\pm,z^\pm]$
\begin{align}
\sigma_{QCA} &= \left(
\begin{array}{c|c}
 z f+f & 0 \\
 0 & z f+f \\
 \hline
 0 & y z-y \\
 x-x z & 0 \\
\end{array}
\right)
+
\sigma_{toric} 
\left(
\begin{array}{cc}
 x z f+f & -f y-f z \\
 0 & 0 \\
\end{array}
\right)
\label{eq:sigmaQCA} \\
&=
\left(
\begin{array}{c|c}
 f y+f x z y+f z-f x z & -f y^2+f y-f z y+f z \\
 -f z x^2-f x+f z x+f & x y f-y f+x z f+f \\
 \hline
 0 & y z-y \\
 x-x z & 0 \\
\end{array}
\right)\label{eq:expandedSigmaQCA}
\end{align}
where $f \in \FF_p^\times$.
The symplectic matrix associated with the TI Clifford QCA
that creates the ground state of our Hamiltonian from that of $H_0 = - \sum_j Z_j + h.c.$,
is 
\begin{align}
&Q_{p,f} = \label{eq:qca}\\
&{\small
\left(
\begin{array}{c|c|c|c}
 \frac{x y z}{4}-\frac{x z}{4}+\frac{z}{4 y}+\frac{1}{4} & \frac{1}{4} x z y^2+\frac{z y}{4}-\frac{x z y}{4}+\frac{z}{4} & f y+f x z y+f z-f x z & -f y^2+f y-f z y+f z \\
 -\frac{z x^2}{4}+\frac{z x}{4 y}+\frac{x}{4}+\frac{1}{4 y} & -\frac{1}{4} y z x^2+\frac{y z x}{4}+\frac{z x}{4}+\frac{1}{4} & -f z x^2-f x+f z x+f & x y f-y f+x z f+f \\
 \hline
 -\frac{x z y}{4 f}-\frac{y}{4 f}+\frac{x z}{4 f}-\frac{1}{4 f} & \frac{x y z}{4 f}-\frac{y}{4 f x} & 0 & y z-y \\
 \frac{x y}{4 f}-\frac{x z}{4 f y} & -\frac{x z y}{4 f}+\frac{y}{4 f}-\frac{x z}{4 f}-\frac{1}{4 f} & x-x z & 0 \\
\end{array}
\right). 
} \nonumber
\end{align}
This matrix means that e.g. the single qudit operator $X$ at the first qudit at coordinate $(0,0,0)$
is mapped to the product of single qudit operators described by the first column.
Each monomial $h x^a y^b z^c$ represents a factor of either $X^h$ or $Z^h$ depending on whether the monomial 
is in the upper or lower half block, supported on the site of coordinate $(a,b,c)$.
Whether the factor acts on the first or second qudit in that site is determined by the row index
within the upper or lower half block.
The first column of $Q$ represents a Pauli operator of $10$, not $14$, nonidentity tensor factors.
The overall phase factor is undetermined but any choice gives a valid Clifford QCA.
Note that $\sigma_{QCA}$ appears in the right two columns of $Q$.
One may check $Q_{p,f}^\dagger \lambda_2 Q_{p,f} = \lambda_2$ directly
or look up the supplementary material.%
\footnote{
Our exposition here is somewhat redundant in view of our use of algebra above with perhaps terse explanation.
The author hopes that a reader who is not too familiar with the polynomial method
can at least understand what the new QCA is.
}

The symplectic matrix $Q_{p,f}$ is in the form $A + z B$ with $A,B$ not involving $z$.
An antihermitian form is
\begin{align}
\Xi_{p,f} =
\left(
\begin{array}{cc}
 \frac{f}{x}-f x & x f+x y f-y f+f \\
 -\frac{f}{x}+\frac{f}{y}-\frac{f}{x y}-f & f y-\frac{f}{y} \\
\end{array}
\right)
\end{align}
whose determinant is $4f^2$.

We have displayed $\sigma_{QCA}$ as in \cref{eq:sigmaQCA}
because it best represents the construction behind it,
which we now explain.

\subsection{Construction by condensation}

The toric code ($\ZZ_p$ gauge theory in 2+1D)~\cite{Kitaev_2003} with odd prime $p$ dimensional qudits
have anyons whose fusion rules follow the group law of $\ZZ_p \times \ZZ_p$.
(The Hamiltonian is depicted in \cref{fig:toricCode}.)
In fact, the fusion, topological spin, and mutual braiding statistics
can be succinctly captured by a quadratic form 
\begin{align}
\Theta_{TC}(a,b) = ab = 
\begin{pmatrix} a & b \end{pmatrix}
\begin{pmatrix} 0 & \half \\ \half & 0 \end{pmatrix}
\begin{pmatrix} a \\ b \end{pmatrix}
\end{align}
on $\FF_p^2$.
Here the first component $a \in \FF_p$ denotes the charge and $b \in \FF_p$ the flux of an anyon.
The value $e^{2\pi i \Theta_{TC} / p}$ is the topological spin of an anyon labeled by $(a,b)$
and the value of the associated symmetric form
\begin{align}
S((a,b),(a',b')) = \Theta_{TC}(a+a',b+b') - \Theta_{TC}(a,b) - \Theta_{TC}(a',b')
\end{align}
gives the mutual statistics between anyons $(a,b)$ and $(a',b')$.
The fusion rule is inscribed in the $\FF_p$-vector space $\FF_p^2$ on which the forms are defined.

It is standard in Witt theory that any quadratic form can be diagonalized
to a direct (orthogonal) sum of one-dimensional forms.
Indeed, for any $f \in \FF_p^\times$
\begin{align}
\begin{pmatrix}1 & f \\ 1 & -f \end{pmatrix}
\begin{pmatrix} 0 & \half \\ \half & 0 \end{pmatrix}
\begin{pmatrix} 1 & 1 \\ f & -f \end{pmatrix} 
=
\begin{pmatrix} f & 0 \\ 0 & -f \end{pmatrix}. \label{eq:basischange}
\end{align}
This means that the anyon theory of the $\ZZ_p$-toric code can be thought of as
a tensor product of two noninteracting anyon theories%
\footnote{
I thank L. Fidkowski for bringing this decomposition to my attention.
}
that are described by
\begin{align}
\Theta_{p,f}(u) = f u^2 \quad \text{ and }  \quad \Theta_{p,-f}(u) = -f u^2,
\end{align}
respectively.
The basis change in \cref{eq:basischange} means that
the one theory with $\Theta_{p,f}$ consists of {\bf dyons} 
that are multiples of the bound state of the unit charge and $f$ flux.
The scalar $f$ is arbitrary but only meaningful up to squares $(\FF_p^\times)^2$,
since squares in $f$ correspond to isomorphic anyon theories.%
\footnote{
Since $\FF_p^\times / (\FF_p^\times)^2$ consists of two classes,
this suggests that the toric code is decomposed in two different ways.
However, this is true only for $p \equiv 1 \mod 4$.
If $p \equiv 3 \mod 4$, then two different choices of $f$ up to squares are $f = \pm 1$.
The different behavior in the two cases is aligned with the Witt group of $\FF_p$.
We will come back to this point later.
}
The dyon theory of $\Theta_{p,-f}$ is the time reversal conjugate of that of $\Theta_{p,f}$.
Below we overload our notation to denote by the symbol $\Theta_{p,\pm f}$ the corresponding dyon theory.

Now we employ an idea in Ref.~\cite[\S IV]{WangSenthil2013}.
Consider an infinite stack of parallel 2D layers of the toric code embedded in 3-dimensional space.
Each layer is viewed as two noninteracting dyon theories $\Theta_{p,f}$ and $\Theta_{p,-f}$.
Let us bind every dyon of $\Theta_{p,f}$ in one layer with its time reversal conjugate of $\Theta_{p,-f}$
in an adjacent, say, upper layer.
The bound state of the two dyons has zero topological spin (boson),
and thus some external interaction may condense these bosons.
In a microscopic lattice Hamiltonian, 
this condensation is implemented by 
turning on a hopping term for the boson and making the strength of that term large.
More concretely,
there are two operators that moves the boson into horizontal and vertical directions.
Each operator is supported across two adjacent layers,
and is written in polynomials in the two columns of the first matrix of \cref{eq:sigmaQCA}.

Deep in the condensed phase,
the ground state should be an eigenstate of its creation and annihilation operators of nonzero eigenvalue,
and in particular all the hopping operators should also have nonzero expectation value.
Then, it is necessary for the hopping operators on the lattice to commute with each other.
The chosen hopping operators in the first matrix of \cref{eq:sigmaQCA} do not commute,
but some modification makes them do.
The modification is the second term in \cref{eq:sigmaQCA},
which was obtained by solving an inhomogeneous linear equation over $\FF_p$.
We have chosen a particular solution to the inhomogeneous linear equation
in order to have a simple expression as in \cref{eq:expandedSigmaQCA}.

Since every dyon has nontrivial mutual statistics with the boson that is condensed,
no layer will have any anyon left after condensation.
This implies that there will be no topological excitation above the condensate,
and the ground state will be nondegenerate.
Indeed, it turns out that 
the terms written as polynomials in \cref{eq:sigmaQCA} suffice to define a commuting Pauli Hamiltonian
whose ground state is nondegenerate on any torus.
A mathematical criterion for this nondegeneracy was established in \cite[Cor.4.2]{Haah2013},
and we have verified the criterion by a direct computation 
that can be found in the supplementary material.
By \cite[Thm.IV.4]{nta3} we obtain a TI Clifford QCA that is shown in \cref{eq:qca}
which creates the ground state of the condensate from a trivial product state.

Note that this construction prefers a particular direction in space.
It is not obvious whether the final result may show emergent isotropy.

\subsection{Nontriviality}

We are going to prove that the TI Clifford QCA $\alpha(p,f)$ 
defined by the symplectic matrix in \cref{eq:qca} is
not equal to a Clifford circuit followed by a shift.
The overall argument is almost the same as that of \cite[Thm.III.16]{nta3}:
We open up a boundary of the 3D system,
the bulk of which is governed by the Hamiltonian of \cref{eq:sigmaQCA}.
We define a ``surface Hamiltonian'' such that the full Hamiltonian on the half space 
satisfies the local topological order condition~\cite{BravyiHastingsMichalakis2010stability};
in colloquial terms, this means that all {\em local} 
degrees of freedom, be it microscopic or emergent, are gapped out.
Since each $\Theta_{p,f}$ of a layer is strongly coupled to $\Theta_{p,-f}$ in the upper layer,
there should be an unpaired $\Theta_{p,f}$ at the top boundary.

If there is a circuit that equals $\alpha_{p,f}^{-1} \otimes \Id$,
then we can disentangle the entire energy spectrum of the bulk Hamiltonian.
Here $\Id$ acts on some ancilla qudits.
Concretely, if the surface is at $z=0$ and the bulk occupies $z \le 0$,
then we may drop all the gates of the hypothetical circuit which are on the region $z > -w$,
where the constant $w$ should be chosen so large that the surface $z=0$ is not acted on by any nonidentity gates.
The action of the truncated circuit $\beta$ agrees with that of $\alpha_{p,f}^{-1} \otimes \Id$
in the deep bulk region $z < -w - r$ where $r$ is the range of the hypothetical circuit.
Deep in the bulk, the Hamiltonian terms are mapped to single site operators $Z$ under $\beta$,
so we can unambiguously define a strictly two-dimensional system living on a slap of thickness $w+r$
where the surface Hamiltonian is left intact.
Further if the hypothetical circuit consisted of Clifford gates that were weakly TI,
then the resulting two-dimensional system would have a commuting Pauli Hamiltonian that is weakly TI.
However, such a Hamiltonian must have a nontrivial boson~\cite[Cor.III.20]{nta3},
which contradicts the fact that $\Theta_{p,f}$ does not have any boson other than the vacuum.
Therefore, the hypothesis that $\alpha_{p,f}$ is a Clifford circuit of a constant depth is disproved.

We will rigorously check that there exists a certain translation invariant surface Hamiltonian
that makes the full system locally topologically ordered~\cite{BravyiHastingsMichalakis2010stability}.
Moreover we will check that every topological charge is represented by some excitation at the surface.
This ensures that the hypothetical circuit does not alter any topological particle content.
These statements depend on the detail of our Hamiltonian including the surface;
unfortunately, we do not have a general way of constructing a surface Hamiltonian that suits our purpose.

Our construction of the surface Hamiltonian will be valid for any $f \in \FF_p^\times$.
Therefore, we will complete the proof of \cref{thm:main} (ii) and (iii).
Indeed, let $g \in \FF_p$ be not a square of any element of $\FF_p$.
(e.g., if $p = 5$ then $g = 2$ qualifies; if $p=3$ then $g = -1$ qualifies.)
Then, the quadratic form $\Theta_{p,1} \oplus \Theta_{p,-g}$ cannot evaluate to zero 
on any input other than the zero.
This implies that the surface Hamiltonian of 
the juxtaposition of two independent systems on which $\alpha_{p,1} \otimes \alpha_{p,-g}$ acts,
does not have any nontrivial boson.
Therefore, $\alpha_{p,1} \otimes \alpha_{p,-g} \not\cong \Id$.
If $p \equiv 1 \mod 4$, then we know $\mathfrak C (D,p)$ has exponent 2,
and $\mathfrak C(D=3,p \equiv 1 \mod 4) \supseteq \ZZ_2 \times \ZZ_2$,
where $\alpha_{p,1}$ and $\alpha_{p,g}$ generate the two factors of $\ZZ_2$, respectively.%
\footnote{
The toric code on qudits of dimension $p \equiv 1 \mod 4$ is thus decomposed in two inequivalent ways:
$\Theta_{p,1} \oplus \Theta_{p,-1} = \Theta_{TC} = \Theta_{p,g} \oplus \Theta_{p,-g}$.
We have $\Theta_{p,1} \cong \Theta_{p,-1}$ and $\Theta_{p,g} \cong \Theta_{p,-g}$ but $\Theta_{p,1} \not\cong \Theta_{p,g}$.
}
If $p \equiv 3 \mod 4$,
then $g = -1 \in \FF_p$ is not a square of any element of $\FF_p$.
Therefore, $\alpha_{p,1} \otimes \alpha_{p,1} \not \cong \Id$,
proving that $\mathfrak C(D, p)$ contains an element $[\alpha]$ of order different from $2$.
Since every element of $\mathfrak C(D, p)$ has exponent 4,
the order of $[\alpha]$ must be $4$.
Therefore, $\mathfrak C(D=3, p \equiv 3 \mod 4)$ contains $\ZZ_4$ generated by $\alpha_{p,1}$.

\begin{remark}
While we consider Clifford circuits in this paper,
let us remark that the QCA $\alpha_{p,f}$ is likely nontrivial 
even if we allow more general quantum circuits in the notion of trivial QCA.
That is, we conjecture that $\alpha_{p,f}$ 
is {\em not} a quantum circuit of a constant depth,
where the local gates are arbitrary and not necessarily Clifford.
This relies on the beliefs that (i) a commuting Hamiltonian cannot realize a chiral theory in 
$2+1$D~\cite{Kitaev_2005},
and further that (ii) any commuting Hamiltonian can only realize an anyon theory 
that is the quantum double (Drinfeld center) of a fusion category.
It is conjectured~\cite{Walker_pc} that a spatially homogeneous commuting Hamiltonian in 2D
is constant depth quantum circuit equivalent to a Levin-Wen model~\cite{Levin_2005}.

Note that the belief (ii) implies that the group of all QCA
modulo constant depth quantum circuits and shifts contains the infinite direct sum
over all prime numbers
\begin{align}
\ZZ_2 
\oplus 
\left(\bigoplus_{p: p \equiv 1 \bmod 4} (\ZZ_2 \oplus \ZZ_2) \right)
\oplus
\left(\bigoplus_{p : p \equiv 3 \bmod 4} \ZZ_4  \right). 
\end{align}
This is because any finite tensor product of dyon theories over distinct primes 
does not contain any boson.
Indeed, a boson is a nonzero solution to an equation
\begin{align}
q_0 \frac{x^2 + xy + y^2}{2} + \sum_{j=1}^n \frac{q_j(z_j)}{p_j} = 0 \in \mathbb Q / \ZZ
\end{align}
where $p_j$ are distinct odd primes; $q_0$ is either $0$ or $1$;
$q_1,\ldots,q_n$ are anisotropic quadratic forms over $\FF_{p_j}$;
$x,y, z_j$ are unknowns.
All $q_j$ are from the examples of this paper, 
and the first $\ZZ_2$ factor corresponding to the prime $2$,
is from \cite{nta3}.
If we multiply $2 p_1 p_2 \cdots p_n$ to the equation
and take reduction modulo $p_j$ that is one of the participating primes,
then we see that every unknown has to be zero modulo $p_j$.
This means that any boson of any finite tensor product of dyon theories 
over distinct primes is the vacuum (anyon).
$\square$
\end{remark}

\begin{remark}
There is a formula~\cite[Prop.6.20]{FroehlichGabbiani1990braid}\cite[App.E]{Kitaev_2005}
\begin{align}
\frac{1}{\mathcal D}\sum_{a} d_a^2 \theta_a  = e^{2\pi i c_- / 8 }, 
\qquad \mathcal D^2 = \sum_b d_b^2 \label{eq:chirality}
\end{align}
where $d_a$ and $\theta_a$ are the quantum dimension and topological spin of anyon $a$, respectively,
and $c_-$ is the chiral central charge.
This relates the anyon theory of the bulk with a property of the edge theory.
Let us evaluate the left-hand side for dyon theories $\Theta_{p,f}$.
Any dyon is an abelian anyon, so $d_a = 1$.
Put 
\begin{align}
F(p,f) = \frac{1}{\sqrt p}\sum_{k=1}^p e^{2\pi i f k^2 / p}.
\end{align} 
Note that $F(p,f) = F(p,f h^2)$ for any $h \in \FF_p^\times$.
If $g$ is a nonsquare element of $\FF_p^\times$ where $p$ is an odd prime,
the two sets $\{ k^2 : k \in \FF_p^\times \}$ and $\{ g k^2 : k \in \FF_p^\times \}$
of the same cardinality partition $\FF_p^\times$. 
Hence,
\begin{align}
F(p,1) + F(p,g) &\propto 2 + \sum_{m=1}^{p-1} \left(e^{2\pi i m^2 / p} + e^{2\pi i g m^2/p}\right)
 = 2 + 2\sum_{m=1}^{p-1} e^{2\pi i m / p} = 0,\\
F(p,1)F(p,1)^* &= \frac 1 p \sum_{a,b=1}^p e^{2\pi i (a^2 - b^2)/p} 
= \frac 1 p \sum_{a,b=1}^p e^{2\pi i ab/p} = 1
\end{align}
where in the second line we use the basis change of \cref{eq:basischange}.
It follows that,
\begin{align}
F(p,1)^* = F(p,-1) = F(p,1) = \pm 1 \text{ if }p \equiv 1 \mod 4 \text{ so that } -1 \in (\FF_p^\times)^2,\\
F(p,1)^* = F(p,-1) = - F(p,1) = \pm i \text{ if } p \equiv 3 \mod 4 \text{ so that } -1 \notin (\FF_p^\times)^2.
\end{align}
Indeed, $F(5,1) = 1$ and one cannot argue that $\alpha_{p=5,f=1}$ is nontrivial
by resorting to \cref{eq:chirality} only.
$\square$
\end{remark}

\subsection{Surface topological order}

Now we have to show that there exists a surface Hamiltonian 
with local topological order~\cite{BravyiHastingsMichalakis2010stability}
and examine the topological particle content to see that the surface Hamiltonian realizes $\Theta_{p,f}$.

The system with the boundary is the lattice $\{ (x,y,z) \in \ZZ^3 : z \le 0 \}$ with two qudits per site.
The (top) boundary is at $z = 0$.
The bulk Hamiltonian is the sum of all terms defined by \cref{eq:sigmaQCA} 
which are supported on the half space $z \le 0$.
We introduce surface terms supported on the boundary $z=0$ by
\begin{align}
\sigma_{bd} =
\left(
\begin{array}{c}
 f x (y-1) y \\
 -f (x-1) x y \\
 (x-1) y \\
 x (y-1) \\
\end{array}
\right).
\end{align}

Our first goal is to show that any finitely supported Pauli operator $O$ 
that commutes with all bulk and all surface terms,
is a product of bulk and surface terms supported within a constant neighborhood of $O$.
To this end,
we observe that the QCA $\alpha_{p,f}$ defined in \cref{eq:qca} maps $X_j,Z_j$ at $z=z_0$ plane
to operators supported on just two planes $z=z_0$ and $z=z_0+1$.
Pretend that $O$ was an operator in $\lat 2 3$ without boundary and look at $\alpha^{-1}_{p,f}(O)$.
Since $\alpha^{-1}_{p,f}(O)$ commutes with every $Z_j$ on a site of $z$-coordinate $\le -1$,
the tensor factors in the half space $z \le -1$ is all $Z$,
and thus can be removed by multiplying by $Z$ there.
So we are reduced to the situation where $\alpha^{-1}_{p,f}(O)$ is supported on $z=0$ plane,
which is equivalent to the situation where $O$ is supported on the top boundary, the $z=0$ plane.
Now, $O$ commutes with the surface terms as well as with the parts $B^{top}$ 
of the bulk terms on the top boundary,
which are represented by polynomials as
\begin{align}
B^{top}=
\text{coefficient of }z\text{ of }\sigma_{QCA} =
\left(
\begin{array}{c|c}
 -x f+x y f+f & f-f y \\
 f x-f x^2 & f x \\
 \hline
 0 & y \\
 -x & 0 \\
\end{array}
\right).
\end{align}
The commutation assumption means that $O$ as a polynomial column matrix
belongs to 
\begin{align}
\ker\left( \begin{pmatrix} B^{top} & \sigma_{bd} \end{pmatrix}^\dagger\lambda_2 \right)
\end{align}
where $\lambda_2$ is the $4$-by-$4$ unit symplectic matrix.
$\begin{pmatrix} B^{top} & \sigma_{bd} \end{pmatrix}$ is a 4-by-3 matrix.
We have used the Buchsbaum-Eisenbud criterion~\cite{BuchsbaumEisenbud1973Exact}
(see the supplementary material)
to check that this kernel is equal to $\im \sigma_{bd}$,
the module that represents the set of all finite products of the surface terms.
Since a membership of a module can be explicitly tested by a division algorithm with a Gr\"obner basis,
any Pauli operator of $\im \sigma_{bd}$ is a product of the generators (the surface terms)
within a constant neighborhood (determined by the degree of the Gr\"obner basis) of the support of the operator.
This is what we have promised to show.
The argument is similar to that presented in \cite[\S III.A.2]{nta3},
but the calculation is more systematic here.

The technical reason we needed to show the local topological order condition
is to apply \cite[Cor.III.20]{nta3} to the two-dimensional system
that would be obtained by a hypothetical circuit matching $\alpha_{p,f}$.
The cited result asserts the existence of a nontrivial boson
if the two-dimensional commuting Pauli Hamiltonian has local topological order
and contains some nontrivial topological charge.

It remains to examine the topological particle content of the system with the top boundary.
By definition, a {\bf topological excitation} for our purpose 
is an equivalence class of finite energy excitations (flipped terms of the Hamiltonian)
modulo locally created ones.
Since we mod out locally created excitations, we may assume that any flipped term
is a surface term;
for any bulk term in our system with the top boundary, the matching element of the flipper 
(the first and second columns of \cref{eq:qca})
is supported in our system with the top boundary,
and hence any flipped bulk term can be eliminated by a local operator.
A pair of surface terms that are flipped can be consolidated (fused)
to a single surface term by the following hopping operators along $x$- and $y$-directions,
\begin{align}
h_x &= \frac{xy}{2} \left(
\begin{array}{c}
 y \\
 1-x \\
 0 \\
 x/f \\
\end{array}
\right),
&h_y &= \frac{xy}{2}\left(
\begin{array}{c}
 y-y^2 \\
 x y-y+1 \\
 -y/f \\
 0 \\
\end{array}
\right), \label{eq:hopping}
\\
\sigma_{bd}^\dagger\lambda_2 h_x &= x-1, 
&\sigma_{bd}^\dagger\lambda_2 h_y &= y-1,
\end{align}
which commute with all the bulk terms
\begin{align}
(B^{top})^\dagger \lambda_2 h_x &= 0,
&(B^{top})^\dagger \lambda_2 h_y &= 0. 
\end{align}
Therefore, a single surface term represents all possible topological charges by its eigenvalue
on an excited state.
\begin{remark}
In fact, the hopping operators generate the commutant $\mathcal A$ of the algebra $\mathcal B$
of all elements of the separator and of the local flipper which are supported on the half space $z \le 0$.
Since each element of the local flipper 
(the first two columns of $Q_{p,f}$ in \cref{eq:qca}) occupies just two layers,
the commutant $\mathcal A$ must be supported on the plane $z = 0$.
By \cref{lem:bdalg}, we know $\mathcal A$ is generated by operators of the columns of $A$ 
in $Q_{p,f} = A + z B$ where $A,B$ do not involve $z$.
The claim that $\mathcal A$ is also generated by the hopping operators,
can be checked either by inspection or by the Buchbaum-Eisenbud criterion
that shows
\begin{align}
\ker \left( \text{coefficient of }z\text{ in } Q_{p,f} ^\dagger \lambda_2 \right) 
= \im \begin{pmatrix} h_x & h_y \end{pmatrix}.
\end{align}
$\square$
\end{remark}

We define the topological spin $\theta_a$ of a topological charge $a$ by
the commutation relation between the hopping operators~\cite{LevinWen2003Fermions}:
\begin{align}
t_1 t_2^\dagger t_3 = \theta_a t_3 t_2^\dagger t_1
\end{align}
where $t_j$ are sufficiently long hopping operators (string operators) 
that move an excitation from a common point (the origin) to locations $p_1,p_2,p_3$
which wind around the common point counterclockwise.
The precise locations of $p_j$ do not affect $\theta_a$.
We choose 
$p_1$ to be a point on the positive $x$-axis,
$p_2$ to be a point on the positive $y$-axis,
and
$p_3$ to be a point on the negative $x$-axis.
In terms of polynomials, the long hopping operators are
conveniently expressed as e.g. $t_1 = (1+x+x^2+ \cdots + x^n)h_x$.
Using these long hopping operators,
we have confirmed that $\theta_a = e^{2\pi i m/p}$ where $m = 1/(4f) \in \FF_p$, independent of $n \ge 1$.
See the supplementary material.
This implies that the system with the boundary has anyons that are captured by a quadratic form $\Theta_{p,1/(4f)}$,
which is congruent to $\Theta_{p,f}$ since the ratio $4f^2$ is a square.
This completes the examination of the surface topological order,
and \cref{thm:main} (ii) and (iii) are now proved.

\section{Discussion}\label{sec:discussion}

We have shown that translation invariant Clifford QCA on systems of prime dimensional qudits 
can be understood via ``small'' groups $\mathfrak C(D,p)$.
Modulo Clifford circuits and shifts, they are trivial in zero, one, and two dimensions ($D=0,1,2$).
In general, at most the fourth power of any TI Clifford QCA is trivial.
The proven exponents $2$ and $4$ cannot be reduced
since we have constructed examples in three dimensions which have true order $2$ and $4$.

However, we do not know how many nontrivial QCA are there in three or higher dimensions.
We expect that the topological order at the boundary of the separator of a Clifford QCA
is simple enough that we can classify them all,
but currently the problem to determine $\mathfrak C(D \ge 3, p)$ remains open.
It is conceivable that the two dimensional surface topological order should be one of $\Theta_{p,f}$,
which would imply
\begin{conjecture}\label{conjecture}
Let $p$ be a prime. Then, $\mathfrak C(D=3,p=2) = \ZZ_2$, 
$\mathfrak C(D=3,p \equiv 1 \mod 4) = \ZZ_2 \times \ZZ_2$,
and
$\mathfrak C(D=3,p \equiv 3 \mod 4) = \ZZ_4$.
\end{conjecture}

We have found boundary algebras useful to show triviality of QCA.
It has become apparent that Witt groups of antihermitian forms is an almost identical subject.
The only piece that stays in the way for the complete equivalence is the notion of ``coarse-graining.''
When we take a smaller subgroup of finite index of the translation group,
the antihermitian form is transformed by a certain injective endomorphism of the base ring.
We do not yet understand how this transformation behaves.
It is plausible that the coarse-graining induces nontrivial endomorphisms on the Witt group.
These endomorphisms constitute a directed system over a single object, the Witt group.
In view of \cref{cor:Xi-to-QCA}, the direct limit is our $\mathfrak C(D,p)$.
This remark goes side-by-side with \cref{app:forms}.

An analogous problem on more general QCA, not necessarily prime qudit Clifford, 
is of course an interesting problem.
Although we have proved that 3D QCA $\alpha_{p,f}$ are nontrivial as Clifford QCA,
the proof is rather improvised.
We introduced the boundary and defined some two-dimensional theory,
which is only physically well motivated but {\em not} necessarily derived.
Only the hopping operators in \cref{eq:hopping} {\em are} derived,
since they are just generators of the boundary algebra.
A proper solution to the classification problem of QCA should include a definition of a better invariant,
which is likely to be in terms of this boundary algebra~\cite{GNVW,FreedmanHastings2019QCA}.
Perhaps, we should phrase the appearance of $\Theta_{p,f}$
as a feature of the boundary algebra that
it contains a locally generated maximal commutative subalgebra 
whose spectrum is not realizable in a disentangled qudit algebra on a plane,
constructed from a direct limit of local matrix algebras.

Our study of antihermitian forms associated with TI Clifford QCA
makes us wonder about an \emph{intrinsic} definition of boundary algebras.
What are analogous conditions for more general von Neumann algebras of quasi-local operators?
Is the notion of visibly simple algebras~\cite{FreedmanHastings2019QCA} enough?
These questions invite a theory of Brauer groups of subalgebras of quasi-local operator algebras.
Trivial elements of this Brauer group should correspond to disentangled algebras.
Such a Brauer group should be at least as rich as the Witt group of finite dimensional antihermitian forms.
Will that group on 2D lattices
match the Witt group of braided fusion categories~\cite{Davydov2010,Davydov2011}?

\appendix

\section{Every QCA in two dimensions is trivial}\label{app:trivialQCA}

In this section, we consider QCA that is not necessarily Clifford or translation invariant.
Recall that a QCA is an automorphism of $*$-algebra of operators on a system of qudits
such that every single qudit operator at $s$ is mapped to an operator supported on an $r$-neighborhood of $s$.
The constant $r$ is called the range and is independent of the system size.
In the main text we considered translation invariant Clifford QCA on infinite lattices directly,
as we did not encounter any subtleties involving infinities.
Actually it was easier for us not to discuss boundary conditions,
which would have complicated the notation without giving any further insight or rigor.
In this section, we resort to (a sequence of) finite systems, 
but are concerned with scaling of the depth of the quantum circuit
(a product of layers of nonoverlapping gates that act on a pair of adjacent qudits)
as a function of the number of qudits or the diameter of the system.
Specifically, we consider a lattice of finitely many qudits on a metric space 
that is a {\em two}-dimensional manifold.
We require that each qudit have a prime dimension upper bounded by a constant,
and there be only a constant number of qudits per unit area.
Here and below, 
whenever we refer to a constant we mean a number that is independent of the system size or the number of qudits in the system.
The assumption of qudit having a prime dimension, which is not necessarily uniform across the lattice,
is just a technical convenience; $\CC^{pq} = \CC^p \otimes \CC^q$ for any positive $p,q \in \ZZ$.
All quantum circuits below are of constant depth.
Every algebra we consider is a $\dagger$-closed finite dimensional complex algebra.

Ref.~\cite{FreedmanHastings2019QCA} shows that any QCA $\alpha$ on 2D is trivial.
The key point in their argument is to construct a QCA $\gamma$ (blending)
that interpolates $\alpha$ on a disk and the identity outside the disk.
The interpolation was made possible by introducing ancillas on the boundary of the disk.
The dimension of the ancilla per unit area is exponential in the diameter of the disk on which $\gamma$ agrees with $\alpha$.
Given this blending, they trivialize the action of $\gamma$ on the disk 
by a constant depth quantum circuit, using another batch of ancillas whose local dimension is exponential in the volume of the disk.
The exponential ancillas are unavoidable in a blending between an arbitrary QCA and the identity;
if $\alpha$ is a parallel shift of all the qudits to the right,
then somewhere near the disk the full flux has to reside.
After all the liberal use of ancillas,
\cite[Thm.3.9]{FreedmanHastings2019QCA} gives an argument for the triviality of 2D QCA which removes the need for large ancillas
by using the blending only for small disks.

Here we revisit the argument for the blending of an arbitrary 2D QCA $\alpha$ into a trivial QCA.
We show that {\em $\alpha$ followed by a shift QCA blends into the identity}.
This simplifies the triviality proof for 2D QCA 
with an extra strength that it requires {\em no} ancillas manifestly.
\begin{theorem}\label{thm:trivial2D}
There is a constant $c \ge 1$ and a real function $d=d(r)$ such that 
for any QCA $\alpha$ of range $r$ on a two-dimensional compact manifold of diameter $L$ with or without boundary 
there exists a shift QCA $\gamma$ of range $\le c r$ such that $\alpha \gamma$ is a quantum circuit (without any ancillas) of depth at most $d(r)$.
\end{theorem}

Our proof of the theorem is along the same line as in \cref{sec:bdalg}.
We partition the space into strips of constant width and color them red and blue alternatingly.
We construct a QCA in each red strip that matches $\alpha$ in the core of the strip up to a shift QCA,
using the decomposition of 1D ``visibly simple algebras'' 
into simple subalgebras that are supported on small disks~\cite[Thm.3.6]{FreedmanHastings2019QCA}
and the Hall marriage theorem~\cite[Chap.22.Thm.3]{BookProof}.
The decomposition of 1D visibly simple algebras is technically the most involved,
and its role is analogous to that of 
the existence of a locally generated maximal commutative subalgebra in any 1D Pauli algebra 
for translation invariant Clifford QCA in the main text.
Since the constructed QCA on a red strip is a 1D QCA,
it is a quantum circuit followed by a shift~\cite{GNVW}.
Applying the inverse of the found circuit to each red strip, 
we are left with blue strips and shifts, with which we can find a circuit and a shift to match, similarly.

Let each qudit $s$ have dimension $D(s)$ a prime,
and let $\mathcal P_s$ be the single-qudit operator algebra of rank $D(s)$.

\begin{lemma}\label{lem:wedderburn}
A simple algebra is isomorphic to a matrix algebra of rank $n$.
If $n$ is a composite number, 
the simple algebra is isomorphic to the tensor product of 
smaller matrix algebras of prime ranks that divide $n$.
The rank of a simple subalgebra divides the rank of the simple algebra into which it is embedded.
\end{lemma}
\begin{proof}
Obvious by the Artin-Wedderburn theorem.
\end{proof}

\begin{lemma}\label{lem:marriage}
Let $\Lambda = \{ s \}$ be a finite set of qudits of prime dimensions,
and let $\mathcal P(\Lambda) = \bigotimes_{s \in \Lambda} \mathcal P_s$ be the operator algebra on $\Lambda$.
Let $\{ \mathcal Q_s \}$ be a collection of mutually commuting simple subalgebras of 
the simple algebra $\mathcal P(\Lambda)$
where the rank of each $\mathcal Q_s$ is a prime.
If $\mathcal Q_s$'s generate $\mathcal P(\Lambda)$,
then there exist an algebra automorphism $\beta$ such that $\mathcal Q_s = \beta(\mathcal P_s)$
where $s$ is a qudit in the support of $\mathcal Q_s$.
\end{lemma}
The argument here is nearly identical to that in the proof of \cite[Lem.II.6]{nta3}.
However, the statement here is sharper and more useful
because the construction of $\beta$ has nothing to do with the geometry of the system of qudits.
The ``locality'' is encoded in the support of $\mathcal Q_s$ but nowhere else.

\begin{proof}
This is an application of the Hall marriage theorem~\cite[Chap.22.Thm.3]{BookProof}
which says that given a collection of subsets $M_s \subseteq M$ of a finite set $M$
if the union of any $k$ subsets from the collection has $\ge k$ elements,
then there is a one-to-one mapping $\{ M_s \} \to M$ 
(i.e., faithful representatives of the subsets).
Indeed, the operator algebra $\mathcal P(\Supp \mathcal Q_s)$ 
on the qudits in the support of $\mathcal Q_s$
contains $\mathcal Q_s$ that is of a prime rank $D(\mathcal Q_s)$.
Hence, by \cref{lem:wedderburn}, 
$\Supp \mathcal Q_s$ contains a $D(\mathcal Q_s)$-dimensional qudit.
Likewise, the union $(\Supp \mathcal Q_{s_1}) \cup \cdots \cup (\Supp \mathcal Q_{s_k})$ 
supports $k$ simple subalgebras $\mathcal Q_{s_j}$,
and hence contains $k$ qudits of dimensions $D(\mathcal Q_{s_1}), \ldots, D(\mathcal Q_{s_k})$.
Let $M_s$ be the set of all qudits of dimension $D(\mathcal Q_s)$ in $\Supp \mathcal Q_s$.
By the Hall marriage theorem,
there is a one-to-one assignment of a $D(\mathcal Q_s)$-dimensional qudit~$s$
for each~$\mathcal Q_s$.
This assignment must be surjective 
since $\prod_s D(\mathcal Q_s) = \prod_s D(P_s)$.
The desired QCA~$\beta$ is defined by 
sending $\mathcal P_s$ to $\mathcal Q_s$ through any isomorphism between the two.
\end{proof}

\begin{proof}[Proof of \cref{thm:trivial2D}]
We first prove the theorem assuming that the underlying manifold is a 2-torus;
this case illustrates the main idea of the proof.
Let the $xy$-plane be the torus 
where every point $(x,y)$ is identified with $(x+L,y)$ and $(x,y+L)$.

Consider a ``red'' horizontal strip $A = A(y_0)$ defined by $y_0 \le y \le y_0 + \ell$,
where $\ell = O(r)$ will be chosen later.
For $s \in A$, the algebra $\mathcal Q_s = \alpha(\mathcal P_s)$ is on the $r$-neighborhood $A^+$ of $A$.
Let an algebra $\mathcal Q(A)$ be generated by all $\mathcal Q_s$ with $s \in A$,
so $\mathcal Q(A)$ is a simple subalgebra of $\mathcal P(A^+)$,
the operator algebra on all qudits in $A^+$.
Since $\mathcal P(A^+)$ is an operator algebra of a one-dimensional array of qudits,
\cite[Thm.3.1,3.2,3.6]{FreedmanHastings2019QCA} says that
the commutant $\mathcal Q'$ of $\mathcal Q(A)$ within $\mathcal P(A^+)$
is generated by simple subalgebras $\mathcal C_j$ of $\mathcal P(A^+)$,
each of which is supported on a disk of radius $r'=O(r)$.
We may assume that each $\mathcal C_j$ is isomorphic to the matrix algebra of a prime rank 
by~\cref{lem:wedderburn}.
Then, the collection of all $\mathcal C_j$ and all $\mathcal Q_s$ with $s \in A$
satisfies the hypothesis of \cref{lem:marriage},
and we obtain a QCA $\beta$ on $A^+$ of range $\max(r,r')$ such that the composition $\beta^{-1}\alpha$ 
maps $\mathcal P_s$ with $s \in A$ to $\mathcal P_{s'}$ for some $s' \in A^+$.
That is, $\beta$ agrees with $\alpha$ on $A$ up to a shift QCA.%
\footnote{
We are using the term ``shift'' QCA slightly more liberally than its name suggests.
We say a QCA $\gamma$ is a shift on a set $S$ of qudits,
if $\gamma(\mathcal P_s)$ is the operator algebra of a qudit for any $s \in S$.
One might insist that a shift QCA must map each operator, say, $Z$ to $Z$,
but we ignore any basis change of a qudit.
}

The constructed $\beta$ is a QCA on $A^+$ that is a one-dimensional array of qudits.
Therefore, $\beta$ itself is a quantum circuit of constant depth followed by a shift~\cite{GNVW}.
Tossing the shift part of $\beta$ away, 
we find a constant depth quantum circuit on $A^+$ whose range is $r'' = O(r')$
that trivializes $\alpha$ on $A$ up to a shift.
The composed QCA of $\alpha$ with the trivializing circuit has range $r''' = \max(r,r',r'')$.
The same construction is applicable for any other strips.
Let us partition the torus into horizontal strips of width $\ell \ge 10r'''$
and color them red and blue alternatingly.
The constant $\ell = O(r''')$ should be tuned so that $L/\ell$ is an integer.
Then, we apply the above construction to modify $\alpha$ by constant depth quantum circuits
to obtain a QCA $\eta$ on the torus of range $r'''$ 
which acts as a shift on the red strips.
The trivializing quantum circuit acts on $r''$-neighborhood of the union of all the red strips.

Next, we focus on {\em one} blue strip $B$;
the other blue strips are treated analogously.
Let $\mathcal P(B)$ be the operator algebra on the qudits in $B$.
An important property of $\eta$ is that $\mathcal Q(B) := \eta(\mathcal P(B))$ has {\em sharp} support;
we say that an algebra $\mathcal A$ has {\bf sharp} support 
if $\mathcal A$ is the operator algebra on all the qudits in $\Supp \mathcal A$.
Indeed, if $T$ is an operator supported on $\Supp \mathcal Q(B)$,
then $\eta^{-1}(T)$ can act neither on any other blue strip because $\ell \gg r'''$
nor on any red strip because for any red strip $A$ where $\eta$ acts as a shift
(the algebra $\eta(\mathcal P(A))$ has sharp support).
Hence, $\eta^{-1}(T)$ is on $B$, and $T \in \mathcal Q_s$.
Therefore, \cref{lem:marriage} applied to $\mathcal Q(B)$
gives a 1D QCA that trivializes $\eta$ on the blue strip.
The theorem for the torus is thus proved.

Before we handle a general 2-manifold $W$,
we note that an algebra's support being sharp is a local property:
we can say that an algebra $\mathcal A$ has sharp support near $s$ (a constant distance neighborhood)
if, for every site $s'$ near $s$, either $s' \notin \Supp \mathcal A$
or $\mathcal P_{s'} \subseteq \mathcal A$.

To handle a general 2-manifold $W$,
consider a disk embedded in $W$
and repeat the above trivialization using
concentric annuli of width $\ell$ that partition the disk.
We obtain a constant depth quantum circuit that modifies the action of $\alpha$ in the interior
(the complement of a constant distance neighborhood of the boundary of the disk)
into that of a shift QCA.
We may choose disks to cover almost all of $W$ except for (the constant distance neighborhood of)
the 1-skeleton $W^1$ in a cell decomposition of $W$.
Let $\alpha_1$ denote the composition of $\alpha$ and the trivializing circuit on the ``big'' disks.
In the complement of $W^1$, $\alpha_1$ acts by shifts.
Now, we take a line segment $e$ ($\simeq$ an interval) in $W^1$,
and consider $\alpha_1( \mathcal P(e))$ (the image of the operator algebra on $e$).
The support of $\mathcal Q(e) = \alpha_1( \mathcal P(e))$ 
is sharp except for two small disks around the end point of $e$.
We can find two simple subalgebras in the commutant of $\mathcal Q(e)$ within 
$\mathcal P(\Supp \mathcal Q(e))$ by the Artin-Wedderburn theorem,
and \cref{lem:marriage} gives a 1D QCA that trivializes $\alpha_1$ on $e$.
Repeating this for every line segment in $W^1$,
we are left with a QCA that is a shift QCA everywhere 
except for small disks centered at the 0-skeleton of~$W$,
which can be easily turned into a shift by a constant depth quantum circuit.
\end{proof}

\section{Digression to automorphism groups of quadratic forms}\label{app:forms}

Over a ring $R$ with an involution $\bar \cdot$,
we consider sesquilinear forms $R^n \times R^n \ni (a,b) \mapsto a^\dagger \lambda b \in R$.
The involution may or may not be trivial.
In our case, the ring is $R = \FF[x_1^\pm,\ldots,x_D^\pm]$ where $\FF$ is a prime field,
the involution is the $\FF$-linear map under which each variable gets inverted,
and $\lambda = \lambda_q$ is the standard symplectic matrix of dimension~$2q$.

The problem on TI Clifford QCA can be cast purely in terms of quadratic forms over Laurent polynomial rings.
Let $Sp^\dagger(q;R)$ be the group of all automorphisms of $\lambda_q$ over $R$,
and let $ESp^\dagger(q;R)$ be the group of all elementary automorphisms of $\lambda_q$ over $R$
generated by elementary automorphisms that are specified in \cref{eq:ElemSymp}.
Let $S_m \subseteq R$ be the subring $\FF[x_1^{\pm m},\ldots,x_D^{\pm m}]$.
(One can also consider subrings specified by finite index subgroups of $\ZZ^D$.)
The collection $\{ Sp^\dagger(q;S_m) \}_{q,m}$ is a directed system.
To increase $q$, we embed $Sp^\dagger(q;S_m)$ into $Sp^\dagger(q+1;S_m)$ in the obvious way,
which corresponds to putting ancilla qudits into the system.
To increase $m$, we embed $Sp^\dagger(q;S_m)$ into $Sp^\dagger (q m' / m;S_{m'})$ whenever $m | m'$
by the base ring change $S_{m'} \ni x_j^{m'} \mapsto (x_j^m)^{m'/m} \in S_m$;
every $S_m$-module becomes an $S_{m'}$-module via $S_{m'} \hookrightarrow S_m$.
The base ring change corresponds to weak translation symmetry breaking.
By taking the direct limit, we obtain $Sp^\dagger(D,\FF)$.
(Strictly speaking, we should have factored out permutation of rows and columns in $Sp^\dagger$ 
to have a direct system because the base ring change does not specify an ordering of bases for free modules;
$S_m$ is an $S_{m'}$-module with a free basis $\{1,x^m,\ldots,(x^m)^{\frac{m'}{m} - 1}\}$.)
Similarly, the collection $\{ ESp^\dagger(q;S_m) \}_{q,m}$ is a directed system,
and gives $ESp^\dagger(D,\FF)$.
The quotient group $Sp^\dagger(D,\FF) / ESp^\dagger(D,\FF)$ is our $\mathfrak C(D,p)$.
It is unknown whether we really need both of the two embeddings;
perhaps, the direct limit obtained by considering only one of the embeddings matches our group $\mathfrak C(D,p)$.
That is, it may be the case that a weak translation symmetry breaking covers the role of ancillas, or vice versa.

We note results in related groups.
In Ref.~\cite{GrunewaldMennickeVaserstein1991SymplecticGroup} 
a similar-looking group $Sp(q;\FF[x_1,\ldots,x_D])$ was shown to be elementarily generated.
Here, unlike our group, the involution is the identity.
This result is extended to that over $\FF[x_1,\ldots,x_s,x_{s+1}^\pm, \ldots,x_D^\pm]$ in Ref.~\cite{Kopeyko1999Symplectic},
where the involution is again the identity.
Hence, the presence of a nontrivial involution on our Laurent polynomial ring makes a substantial difference.

\section{Equivalence of two QCA on qubits ($p=2$)}

In Ref.~\cite{nta3} we have constructed a $D=3$ QCA by disentangling a Walker-Wang model~\cite{Walker_2011,BCFV} for 3-fermion theory.
The idea of coupled layers~\cite{WangSenthil2013} is originally explained for a layer of doubled toric codes over qubits ($p=2$)
such that an exposed surface of a 3-dimensional bulk supports the 3-fermion theory.
The two constructions should give the same theory on intuitive grounds,
and we have explicitly checked that this is the case.
In Supplementary material,
we have realized the coupled layer construction and computed the antihermitian form
to find that
\begin{align}
\Xi_{p=2} = \left(
\begin{array}{cccc}
 x+\frac{1}{x} & 1 & 0 & (1+x)(1+y) \\
 1 & y+\frac{1}{y} &(1+x)(1+y) & 0 \\
 0 & (1+\frac 1 x)(1+\frac 1 y) & x+\frac{1}{x} & 1 \\
 (1+\frac 1 x)(1+\frac 1 y) & 0 & 1 & y+\frac{1}{y} \\
\end{array}
\right).
\end{align}
We have also computed the antihermitian form of the QCA in \cite{nta3}, 
and found that they are the same.
By \cref{cor:Xi-to-QCA}, the two QCA are equivalent.
This is consistent with \cref{conjecture}.

\section*{Supplementary Material}

We provide a calculation note written in Mathematica.
There are several sections in the Mathematica notebook,
containing straightforward but laborious calculation using Laurent polynomials.
The first section contains the construction of Pauli stabilizer Hamiltonians on $D=3$ lattice,
realizing certain Walker--Wang models.
The second section contains calculation of disentangling QCA of the Hamiltonians and boundary antihermitian forms.
The third section shows the surface Hamiltonians and their topological particle content.
The fourth section simplifies a nontrivial QCA over qubits ($\CC^2$)
which was first reported in~\cite{nta3}.

\begin{acknowledgments}
I thank Lukasz Fidkowski, Mike Freedman, and Matt Hastings for useful discussions,
and Dominic Williamson for Refs.~\cite{Davydov2010,Davydov2011}.
\end{acknowledgments}

\bibliography{clifqca-ref}

%merlin.mbs apsrev4-1.bst 2010-07-25 4.21a (PWD, AO, DPC) hacked
%Control: key (0)
%Control: author (0) dotless jnrlst
%Control: editor formatted (1) identically to author
%Control: production of article title (0) allowed
%Control: page (1) range
%Control: year (0) verbatim
%Control: production of eprint (0) enabled
\begin{thebibliography}{36}%
\makeatletter
\providecommand \@ifxundefined [1]{%
 \@ifx{#1\undefined}
}%
\providecommand \@ifnum [1]{%
 \ifnum #1\expandafter \@firstoftwo
 \else \expandafter \@secondoftwo
 \fi
}%
\providecommand \@ifx [1]{%
 \ifx #1\expandafter \@firstoftwo
 \else \expandafter \@secondoftwo
 \fi
}%
\providecommand \natexlab [1]{#1}%
\providecommand \enquote  [1]{``#1''}%
\providecommand \bibnamefont  [1]{#1}%
\providecommand \bibfnamefont [1]{#1}%
\providecommand \citenamefont [1]{#1}%
\providecommand \href@noop [0]{\@secondoftwo}%
\providecommand \href [0]{\begingroup \@sanitize@url \@href}%
\providecommand \@href[1]{\@@startlink{#1}\@@href}%
\providecommand \@@href[1]{\endgroup#1\@@endlink}%
\providecommand \@sanitize@url [0]{\catcode `\\12\catcode `\$12\catcode
  `\&12\catcode `\#12\catcode `\^12\catcode `\_12\catcode `\%12\relax}%
\providecommand \@@startlink[1]{}%
\providecommand \@@endlink[0]{}%
\providecommand \url  [0]{\begingroup\@sanitize@url \@url }%
\providecommand \@url [1]{\endgroup\@href {#1}{\urlprefix }}%
\providecommand \urlprefix  [0]{URL }%
\providecommand \Eprint [0]{\href }%
\providecommand \doibase [0]{http://dx.doi.org/}%
\providecommand \selectlanguage [0]{\@gobble}%
\providecommand \bibinfo  [0]{\@secondoftwo}%
\providecommand \bibfield  [0]{\@secondoftwo}%
\providecommand \translation [1]{[#1]}%
\providecommand \BibitemOpen [0]{}%
\providecommand \bibitemStop [0]{}%
\providecommand \bibitemNoStop [0]{.\EOS\space}%
\providecommand \EOS [0]{\spacefactor3000\relax}%
\providecommand \BibitemShut  [1]{\csname bibitem#1\endcsname}%
\let\auto@bib@innerbib\@empty
%</preamble>
\bibitem [{\citenamefont {{Watrous}}(1995)}]{Watrous1995}%
  \BibitemOpen
  \bibfield  {author} {\bibinfo {author} {\bibfnamefont {J.}~\bibnamefont
  {{Watrous}}},\ }\bibfield  {title} {\enquote {\bibinfo {title} {On
  one-dimensional quantum cellular automata},}\ }in\ \href {\doibase
  10.1109/SFCS.1995.492583} {\emph {\bibinfo {booktitle} {Proceedings of IEEE
  36th Annual Foundations of Computer Science}}}\ (\bibinfo {year} {1995})\
  pp.\ \bibinfo {pages} {528--537}\BibitemShut {NoStop}%
\bibitem [{\citenamefont {Schumacher}\ and\ \citenamefont
  {Werner}()}]{SchumacherWerner2004}%
  \BibitemOpen
  \bibfield  {author} {\bibinfo {author} {\bibfnamefont {B.}~\bibnamefont
  {Schumacher}}\ and\ \bibinfo {author} {\bibfnamefont {R.~F.}\ \bibnamefont
  {Werner}},\ }\bibfield  {title} {\enquote {\bibinfo {title} {Reversible
  quantum cellular automata},}\ }\href@noop {} {\ }\Eprint
  {http://arxiv.org/abs/quant-ph/0405174} {arXiv:quant-ph/0405174} \BibitemShut
  {NoStop}%
\bibitem [{\citenamefont {Gross}\ \emph {et~al.}(2012)\citenamefont {Gross},
  \citenamefont {Nesme}, \citenamefont {Vogts},\ and\ \citenamefont
  {Werner}}]{GNVW}%
  \BibitemOpen
  \bibfield  {author} {\bibinfo {author} {\bibfnamefont {D.}~\bibnamefont
  {Gross}}, \bibinfo {author} {\bibfnamefont {V.}~\bibnamefont {Nesme}},
  \bibinfo {author} {\bibfnamefont {H.}~\bibnamefont {Vogts}}, \ and\ \bibinfo
  {author} {\bibfnamefont {R.~F.}\ \bibnamefont {Werner}},\ }\bibfield  {title}
  {\enquote {\bibinfo {title} {Index theory of one dimensional quantum walks
  and cellular automata},}\ }\href {\doibase 10.1007/s00220-012-1423-1}
  {\bibfield  {journal} {\bibinfo  {journal} {Communications in Mathematical
  Physics}\ }\textbf {\bibinfo {volume} {310}},\ \bibinfo {pages} {419--454}
  (\bibinfo {year} {2012})},\ \Eprint {http://arxiv.org/abs/0910.3675}
  {arXiv:0910.3675} \BibitemShut {NoStop}%
\bibitem [{\citenamefont {Schlingemann}\ \emph {et~al.}(2008)\citenamefont
  {Schlingemann}, \citenamefont {Vogts},\ and\ \citenamefont
  {Werner}}]{clifQCA}%
  \BibitemOpen
  \bibfield  {author} {\bibinfo {author} {\bibfnamefont {Dirk-M.}\ \bibnamefont
  {Schlingemann}}, \bibinfo {author} {\bibfnamefont {Holger}\ \bibnamefont
  {Vogts}}, \ and\ \bibinfo {author} {\bibfnamefont {Reinhard~F.}\ \bibnamefont
  {Werner}},\ }\bibfield  {title} {\enquote {\bibinfo {title} {On the structure
  of clifford quantum cellular automata},}\ }\href {\doibase 10.1063/1.3005565}
  {\bibfield  {journal} {\bibinfo  {journal} {Journal of Mathematical Physics}\
  }\textbf {\bibinfo {volume} {49}},\ \bibinfo {pages} {112104} (\bibinfo
  {year} {2008})},\ \Eprint {http://arxiv.org/abs/0804.4447} {arXiv:0804.4447}
  \BibitemShut {NoStop}%
\bibitem [{\citenamefont {Freedman}\ and\ \citenamefont
  {Hastings}(2020)}]{FreedmanHastings2019QCA}%
  \BibitemOpen
  \bibfield  {author} {\bibinfo {author} {\bibfnamefont {M.}~\bibnamefont
  {Freedman}}\ and\ \bibinfo {author} {\bibfnamefont {M.~B.}\ \bibnamefont
  {Hastings}},\ }\bibfield  {title} {\enquote {\bibinfo {title} {Classification
  of quantum cellular automata},}\ }\href {\doibase 10.1007/s00220-020-03735-y}
  {\bibfield  {journal} {\bibinfo  {journal} {Commun. Math. Phys.}\ }\textbf
  {\bibinfo {volume} {376}},\ \bibinfo {pages} {1171--1222} (\bibinfo {year}
  {2020})},\ \Eprint {http://arxiv.org/abs/1902.10285} {arXiv:1902.10285}
  \BibitemShut {NoStop}%
\bibitem [{\citenamefont {Haah}\ \emph {et~al.}()\citenamefont {Haah},
  \citenamefont {Fidkowski},\ and\ \citenamefont {Hastings}}]{nta3}%
  \BibitemOpen
  \bibfield  {author} {\bibinfo {author} {\bibfnamefont {Jeongwan}\
  \bibnamefont {Haah}}, \bibinfo {author} {\bibfnamefont {Lukasz}\ \bibnamefont
  {Fidkowski}}, \ and\ \bibinfo {author} {\bibfnamefont {Matthew~B.}\
  \bibnamefont {Hastings}},\ }\bibfield  {title} {\enquote {\bibinfo {title}
  {Nontrivial quantum cellular automata in higher dimensions},}\ }\href@noop {}
  {\ }\Eprint {http://arxiv.org/abs/1812.01625} {arXiv:1812.01625} \BibitemShut
  {NoStop}%
\bibitem [{\citenamefont {Arrighi}\ \emph {et~al.}(2011)\citenamefont
  {Arrighi}, \citenamefont {Nesme},\ and\ \citenamefont
  {Werner}}]{Arrighi2007}%
  \BibitemOpen
  \bibfield  {author} {\bibinfo {author} {\bibfnamefont {Pablo}\ \bibnamefont
  {Arrighi}}, \bibinfo {author} {\bibfnamefont {Vincent}\ \bibnamefont
  {Nesme}}, \ and\ \bibinfo {author} {\bibfnamefont {Reinhard}\ \bibnamefont
  {Werner}},\ }\bibfield  {title} {\enquote {\bibinfo {title} {Unitarity plus
  causality implies localizability},}\ }\href {\doibase
  10.1016/j.jcss.2010.05.004} {\bibfield  {journal} {\bibinfo  {journal}
  {Journal of Computer and System Sciences}\ }\textbf {\bibinfo {volume}
  {77}},\ \bibinfo {pages} {372--378} (\bibinfo {year} {2011})},\ \Eprint
  {http://arxiv.org/abs/0711.3975} {arXiv:0711.3975} \BibitemShut {NoStop}%
\bibitem [{\citenamefont {Freedman}\ \emph {et~al.}(2022)\citenamefont
  {Freedman}, \citenamefont {Haah},\ and\ \citenamefont {Hastings}}]{FHH2019}%
  \BibitemOpen
  \bibfield  {author} {\bibinfo {author} {\bibfnamefont {Michael}\ \bibnamefont
  {Freedman}}, \bibinfo {author} {\bibfnamefont {Jeongwan}\ \bibnamefont
  {Haah}}, \ and\ \bibinfo {author} {\bibfnamefont {Matthew~B.}\ \bibnamefont
  {Hastings}},\ }\bibfield  {title} {\enquote {\bibinfo {title} {The group
  structure of quantum cellular automata},}\ }\href {\doibase
  10.1007/s00220-022-04316-x} {\bibfield  {journal} {\bibinfo  {journal}
  {Commun. Math. Phys.}\ }\textbf {\bibinfo {volume} {389}},\ \bibinfo {pages}
  {1277--1302} (\bibinfo {year} {2022})},\ \Eprint
  {http://arxiv.org/abs/1910.07998} {arXiv:1910.07998} \BibitemShut {NoStop}%
\bibitem [{\citenamefont {Haah}(2013)}]{Haah2013}%
  \BibitemOpen
  \bibfield  {author} {\bibinfo {author} {\bibfnamefont {Jeongwan}\
  \bibnamefont {Haah}},\ }\bibfield  {title} {\enquote {\bibinfo {title}
  {Commuting pauli hamiltonians as maps between free modules},}\ }\href
  {\doibase 10.1007/s00220-013-1810-2} {\bibfield  {journal} {\bibinfo
  {journal} {Commun. Math. Phys.}\ }\textbf {\bibinfo {volume} {324}},\
  \bibinfo {pages} {351--399} (\bibinfo {year} {2013})},\ \Eprint
  {http://arxiv.org/abs/1204.1063} {arXiv:1204.1063} \BibitemShut {NoStop}%
\bibitem [{\citenamefont {Suslin}(1977)}]{Suslin1977Stability}%
  \BibitemOpen
  \bibfield  {author} {\bibinfo {author} {\bibfnamefont {A.~A.}\ \bibnamefont
  {Suslin}},\ }\bibfield  {title} {\enquote {\bibinfo {title} {On the structure
  of the special linear group over polynomial rings},}\ }\href {\doibase
  10.1070/IM1977v011n02ABEH001709} {\bibfield  {journal} {\bibinfo  {journal}
  {Mathematics of the USSR-Izvestiya}\ }\textbf {\bibinfo {volume} {11}},\
  \bibinfo {pages} {221} (\bibinfo {year} {1977})}\BibitemShut {NoStop}%
\bibitem [{\citenamefont {Wang}\ and\ \citenamefont
  {Senthil}(2013)}]{WangSenthil2013}%
  \BibitemOpen
  \bibfield  {author} {\bibinfo {author} {\bibfnamefont {Chong}\ \bibnamefont
  {Wang}}\ and\ \bibinfo {author} {\bibfnamefont {T.}~\bibnamefont {Senthil}},\
  }\bibfield  {title} {\enquote {\bibinfo {title} {Boson topological
  insulators: A window into highly entangled quantum phases},}\ }\href
  {\doibase 10.1103/PhysRevB.87.235122} {\bibfield  {journal} {\bibinfo
  {journal} {journaltitle = {Phys. Rev. B 87, 235122 (2013)},}\ }\textbf
  {\bibinfo {volume} {87}},\ \bibinfo {pages} {235122} (\bibinfo {year}
  {2013})},\ \Eprint {http://arxiv.org/abs/1302.6234} {arXiv:1302.6234}
  \BibitemShut {NoStop}%
\bibitem [{\citenamefont {Milnor}\ and\ \citenamefont
  {Husemoller}(1973)}]{MilnorHusemoller}%
  \BibitemOpen
  \bibfield  {author} {\bibinfo {author} {\bibfnamefont {John}\ \bibnamefont
  {Milnor}}\ and\ \bibinfo {author} {\bibfnamefont {Dale}\ \bibnamefont
  {Husemoller}},\ }\href {\doibase 10.1007/978-3-642-88330-9} {\emph {\bibinfo
  {title} {Symmetric Bilinear Forms}}}\ (\bibinfo  {publisher} {Springer,
  Berlin, Heidelberg},\ \bibinfo {year} {1973})\BibitemShut {NoStop}%
\bibitem [{\citenamefont {Lam}(2004)}]{Lam}%
  \BibitemOpen
  \bibfield  {author} {\bibinfo {author} {\bibfnamefont {T.~Y.}\ \bibnamefont
  {Lam}},\ }\href@noop {} {\emph {\bibinfo {title} {Introduction to Quadratic
  Forms over Fields}}},\ Graduate studies in mathematics 67\ (\bibinfo
  {publisher} {American Mathematical Society},\ \bibinfo {year}
  {2004})\BibitemShut {NoStop}%
\bibitem [{\citenamefont {Calderbank}\ \emph {et~al.}(1997)\citenamefont
  {Calderbank}, \citenamefont {Rains}, \citenamefont {Shor},\ and\
  \citenamefont {Sloane}}]{CalderbankRainsShorEtAl1997Quantum}%
  \BibitemOpen
  \bibfield  {author} {\bibinfo {author} {\bibfnamefont {A.~R.}\ \bibnamefont
  {Calderbank}}, \bibinfo {author} {\bibfnamefont {E.~M}\ \bibnamefont
  {Rains}}, \bibinfo {author} {\bibfnamefont {P.~W.}\ \bibnamefont {Shor}}, \
  and\ \bibinfo {author} {\bibfnamefont {N.~J.~A.}\ \bibnamefont {Sloane}},\
  }\bibfield  {title} {\enquote {\bibinfo {title} {Quantum error correction and
  orthogonal geometry},}\ }\href {\doibase 10.1103/PhysRevLett.78.405}
  {\bibfield  {journal} {\bibinfo  {journal} {Phys. Rev. Lett.}\ }\textbf
  {\bibinfo {volume} {78}},\ \bibinfo {pages} {405--408} (\bibinfo {year}
  {1997})},\ \Eprint {http://arxiv.org/abs/quant-ph/9605005}
  {arXiv:quant-ph/9605005} \BibitemShut {NoStop}%
\bibitem [{\citenamefont {Haah}\ \emph {et~al.}(2017)\citenamefont {Haah},
  \citenamefont {Hastings}, \citenamefont {Poulin},\ and\ \citenamefont
  {Wecker}}]{HHPW2017}%
  \BibitemOpen
  \bibfield  {author} {\bibinfo {author} {\bibfnamefont {Jeongwan}\
  \bibnamefont {Haah}}, \bibinfo {author} {\bibfnamefont {Matthew~B.}\
  \bibnamefont {Hastings}}, \bibinfo {author} {\bibfnamefont {D.}~\bibnamefont
  {Poulin}}, \ and\ \bibinfo {author} {\bibfnamefont {D.}~\bibnamefont
  {Wecker}},\ }\bibfield  {title} {\enquote {\bibinfo {title} {Magic state
  distillation with low space overhead and optimal asymptotic input count},}\
  }\href {\doibase 10.22331/q-2017-10-03-31} {\bibfield  {journal} {\bibinfo
  {journal} {Quantum}\ }\textbf {\bibinfo {volume} {1}},\ \bibinfo {pages} {31}
  (\bibinfo {year} {2017})},\ \Eprint {http://arxiv.org/abs/1703.07847}
  {arXiv:1703.07847} \BibitemShut {NoStop}%
\bibitem [{\citenamefont {Eisenbud}(2004)}]{Eisenbud}%
  \BibitemOpen
  \bibfield  {author} {\bibinfo {author} {\bibfnamefont {David}\ \bibnamefont
  {Eisenbud}},\ }\href@noop {} {\emph {\bibinfo {title} {Commutative Algebra
  with a View Toward Algebraic Geometry}}}\ (\bibinfo  {publisher} {Springer},\
  \bibinfo {year} {2004})\BibitemShut {NoStop}%
\bibitem [{\citenamefont {Swan}(1978)}]{Swan1978}%
  \BibitemOpen
  \bibfield  {author} {\bibinfo {author} {\bibfnamefont {Richard~G.}\
  \bibnamefont {Swan}},\ }\bibfield  {title} {\enquote {\bibinfo {title}
  {Projective modules over {Laurent} polynomial rings},}\ }\href {\doibase
  10.2307/1997613} {\bibfield  {journal} {\bibinfo  {journal} {Transactions of
  the American Mathematical Society}\ }\textbf {\bibinfo {volume} {237}},\
  \bibinfo {pages} {111--120} (\bibinfo {year} {1978})}\BibitemShut {NoStop}%
\bibitem [{\citenamefont {Lam}(2006)}]{LamSerreBook}%
  \BibitemOpen
  \bibfield  {author} {\bibinfo {author} {\bibfnamefont {Tsit~Yuen}\
  \bibnamefont {Lam}},\ }\href@noop {} {\emph {\bibinfo {title} {Serre's
  problem on projective modules}}}\ (\bibinfo  {publisher} {Springer},\
  \bibinfo {year} {2006})\BibitemShut {NoStop}%
\bibitem [{\citenamefont {Elman}\ \emph {et~al.}(2008)\citenamefont {Elman},
  \citenamefont {Karpenko},\ and\ \citenamefont {Merkurjev}}]{Kniga}%
  \BibitemOpen
  \bibfield  {author} {\bibinfo {author} {\bibfnamefont {Richard}\ \bibnamefont
  {Elman}}, \bibinfo {author} {\bibfnamefont {Nikita}\ \bibnamefont
  {Karpenko}}, \ and\ \bibinfo {author} {\bibfnamefont {Alexander}\
  \bibnamefont {Merkurjev}},\ }\href
  {https://sites.ualberta.ca/~karpenko/publ/Kniga.pdf} {\emph {\bibinfo {title}
  {The Algebraic and Geometric Theory of Quadratic Forms}}},\ \bibinfo {series}
  {Colloquium Publications}, Vol.~\bibinfo {volume} {56}\ (\bibinfo
  {publisher} {American Mathematical Society},\ \bibinfo {year}
  {2008})\BibitemShut {NoStop}%
\bibitem [{\citenamefont {Lang}(2002)}]{Lang}%
  \BibitemOpen
  \bibfield  {author} {\bibinfo {author} {\bibfnamefont {Serge}\ \bibnamefont
  {Lang}},\ }\href@noop {} {\emph {\bibinfo {title} {Algebra}}},\ \bibinfo
  {edition} {revised 3rd}\ ed.\ (\bibinfo  {publisher} {Springer},\ \bibinfo
  {year} {2002})\BibitemShut {NoStop}%
\bibitem [{\citenamefont {Cohn}(1966)}]{Cohn1966}%
  \BibitemOpen
  \bibfield  {author} {\bibinfo {author} {\bibfnamefont {Paul~M.}\ \bibnamefont
  {Cohn}},\ }\bibfield  {title} {\enquote {\bibinfo {title} {On the structure
  of the {$GL_2$} of a ring},}\ }\href
  {http://www.numdam.org/item/PMIHES_1966__30__5_0} {\bibfield  {journal}
  {\bibinfo  {journal} {Publications Math\'ematiques de l'IH\'ES}\ }\textbf
  {\bibinfo {volume} {30}},\ \bibinfo {pages} {5--53} (\bibinfo {year}
  {1966})}\BibitemShut {NoStop}%
\bibitem [{\citenamefont {Kitaev}(2003)}]{Kitaev_2003}%
  \BibitemOpen
  \bibfield  {author} {\bibinfo {author} {\bibfnamefont {A.Yu.}\ \bibnamefont
  {Kitaev}},\ }\bibfield  {title} {\enquote {\bibinfo {title} {Fault-tolerant
  quantum computation by anyons},}\ }\href {\doibase
  10.1016/s0003-4916(02)00018-0} {\bibfield  {journal} {\bibinfo  {journal}
  {Annals of Physics}\ }\textbf {\bibinfo {volume} {303}},\ \bibinfo {pages}
  {2--30} (\bibinfo {year} {2003})},\ \Eprint
  {http://arxiv.org/abs/quant-ph/9707021} {arXiv:quant-ph/9707021} \BibitemShut
  {NoStop}%
\bibitem [{\citenamefont {Bravyi}\ \emph {et~al.}(2010)\citenamefont {Bravyi},
  \citenamefont {Hastings},\ and\ \citenamefont
  {Michalakis}}]{BravyiHastingsMichalakis2010stability}%
  \BibitemOpen
  \bibfield  {author} {\bibinfo {author} {\bibfnamefont {Sergey}\ \bibnamefont
  {Bravyi}}, \bibinfo {author} {\bibfnamefont {Matthew}\ \bibnamefont
  {Hastings}}, \ and\ \bibinfo {author} {\bibfnamefont {Spyridon}\ \bibnamefont
  {Michalakis}},\ }\bibfield  {title} {\enquote {\bibinfo {title} {Topological
  quantum order: Stability under local perturbations},}\ }\href {\doibase
  10.1063/1.3490195} {\bibfield  {journal} {\bibinfo  {journal} {J. Math.
  Phys.}\ }\textbf {\bibinfo {volume} {51}},\ \bibinfo {pages} {093512}
  (\bibinfo {year} {2010})},\ \Eprint {http://arxiv.org/abs/1001.0344}
  {arXiv:1001.0344} \BibitemShut {NoStop}%
\bibitem [{\citenamefont {{Kitaev}}(2006)}]{Kitaev_2005}%
  \BibitemOpen
  \bibfield  {author} {\bibinfo {author} {\bibfnamefont {A.}~\bibnamefont
  {{Kitaev}}},\ }\bibfield  {title} {\enquote {\bibinfo {title} {{Anyons in an
  exactly solved model and beyond}},}\ }\href {\doibase
  10.1016/j.aop.2005.10.005} {\bibfield  {journal} {\bibinfo  {journal} {Annals
  of Physics}\ }\textbf {\bibinfo {volume} {321}},\ \bibinfo {pages} {2--111}
  (\bibinfo {year} {2006})},\ \Eprint {http://arxiv.org/abs/cond-mat/0506438}
  {arXiv:cond-mat/0506438} \BibitemShut {NoStop}%
\bibitem [{\citenamefont {Walker}(2019)}]{Walker_pc}%
  \BibitemOpen
  \bibfield  {author} {\bibinfo {author} {\bibfnamefont {Kevin}\ \bibnamefont
  {Walker}},\ }\href@noop {} {} (\bibinfo {year} {2019}),\ \bibinfo {note}
  {private communication}\BibitemShut {NoStop}%
\bibitem [{\citenamefont {Levin}\ and\ \citenamefont {Wen}(2005)}]{Levin_2005}%
  \BibitemOpen
  \bibfield  {author} {\bibinfo {author} {\bibfnamefont {Michael~A.}\
  \bibnamefont {Levin}}\ and\ \bibinfo {author} {\bibfnamefont {Xiao-Gang}\
  \bibnamefont {Wen}},\ }\bibfield  {title} {\enquote {\bibinfo {title}
  {String-net condensation:{\hspace{1em}}a physical mechanism for topological
  phases},}\ }\href {\doibase 10.1103/physrevb.71.045110} {\bibfield  {journal}
  {\bibinfo  {journal} {Physical Review B}\ }\textbf {\bibinfo {volume} {71}},\
  \bibinfo {pages} {045110} (\bibinfo {year} {2005})},\ \Eprint
  {http://arxiv.org/abs/cond-mat/0404617} {arXiv:cond-mat/0404617} \BibitemShut
  {NoStop}%
\bibitem [{\citenamefont {Fr\"ohlich}\ and\ \citenamefont
  {Gabbiani}(1990)}]{FroehlichGabbiani1990braid}%
  \BibitemOpen
  \bibfield  {author} {\bibinfo {author} {\bibfnamefont {J.}~\bibnamefont
  {Fr\"ohlich}}\ and\ \bibinfo {author} {\bibfnamefont {F.}~\bibnamefont
  {Gabbiani}},\ }\bibfield  {title} {\enquote {\bibinfo {title} {Braid
  statistics in local quantum theory},}\ }\href {\doibase
  10.1142/S0129055X90000107} {\bibfield  {journal} {\bibinfo  {journal}
  {Reviews in Mathematical Physics}\ }\textbf {\bibinfo {volume} {02}},\
  \bibinfo {pages} {251--353} (\bibinfo {year} {1990})}\BibitemShut {NoStop}%
\bibitem [{\citenamefont {Buchsbaum}\ and\ \citenamefont
  {Eisenbud}(1973)}]{BuchsbaumEisenbud1973Exact}%
  \BibitemOpen
  \bibfield  {author} {\bibinfo {author} {\bibfnamefont {David~A}\ \bibnamefont
  {Buchsbaum}}\ and\ \bibinfo {author} {\bibfnamefont {David}\ \bibnamefont
  {Eisenbud}},\ }\bibfield  {title} {\enquote {\bibinfo {title} {What makes a
  complex exact?}}\ }\href {\doibase 10.1016/0021-8693(73)90044-6} {\bibfield
  {journal} {\bibinfo  {journal} {Journal of Algebra}\ }\textbf {\bibinfo
  {volume} {25}},\ \bibinfo {pages} {259--268} (\bibinfo {year}
  {1973})}\BibitemShut {NoStop}%
\bibitem [{\citenamefont {Levin}\ and\ \citenamefont
  {Wen}(2003)}]{LevinWen2003Fermions}%
  \BibitemOpen
  \bibfield  {author} {\bibinfo {author} {\bibfnamefont {Michael}\ \bibnamefont
  {Levin}}\ and\ \bibinfo {author} {\bibfnamefont {Xiao-Gang}\ \bibnamefont
  {Wen}},\ }\bibfield  {title} {\enquote {\bibinfo {title} {Fermions, strings,
  and gauge fields in lattice spin models},}\ }\href {\doibase
  10.1103/PhysRevB.67.245316} {\bibfield  {journal} {\bibinfo  {journal} {Phys.
  Rev. B}\ }\textbf {\bibinfo {volume} {67}},\ \bibinfo {pages} {245316}
  (\bibinfo {year} {2003})},\ \Eprint {http://arxiv.org/abs/cond-mat/0302460}
  {arXiv:cond-mat/0302460} \BibitemShut {NoStop}%
\bibitem [{\citenamefont {Davydov}\ \emph
  {et~al.}(2013{\natexlab{a}})\citenamefont {Davydov}, \citenamefont {Mueger},
  \citenamefont {Nikshych},\ and\ \citenamefont {Ostrik}}]{Davydov2010}%
  \BibitemOpen
  \bibfield  {author} {\bibinfo {author} {\bibfnamefont {Alexei}\ \bibnamefont
  {Davydov}}, \bibinfo {author} {\bibfnamefont {Michael}\ \bibnamefont
  {Mueger}}, \bibinfo {author} {\bibfnamefont {Dmitri}\ \bibnamefont
  {Nikshych}}, \ and\ \bibinfo {author} {\bibfnamefont {Victor}\ \bibnamefont
  {Ostrik}},\ }\bibfield  {title} {\enquote {\bibinfo {title} {The {Witt} group
  of non-degenerate braided fusion categories},}\ }\href {\doibase
  10.1515/crelle.2012.014} {\bibfield  {journal} {\bibinfo  {journal} {Journal
  f\"ur die reine und angewandte {M}athematik}\ }\textbf {\bibinfo {volume}
  {677}},\ \bibinfo {pages} {135--177} (\bibinfo {year}
  {2013}{\natexlab{a}})},\ \Eprint {http://arxiv.org/abs/1009.2117}
  {arXiv:1009.2117} \BibitemShut {NoStop}%
\bibitem [{\citenamefont {Davydov}\ \emph
  {et~al.}(2013{\natexlab{b}})\citenamefont {Davydov}, \citenamefont
  {Nikshych},\ and\ \citenamefont {Ostrik}}]{Davydov2011}%
  \BibitemOpen
  \bibfield  {author} {\bibinfo {author} {\bibfnamefont {Alexei}\ \bibnamefont
  {Davydov}}, \bibinfo {author} {\bibfnamefont {Dmitri}\ \bibnamefont
  {Nikshych}}, \ and\ \bibinfo {author} {\bibfnamefont {Victor}\ \bibnamefont
  {Ostrik}},\ }\bibfield  {title} {\enquote {\bibinfo {title} {On the structure
  of the {Witt} group of braided fusion categories},}\ }\href {\doibase
  10.1007/s00029-012-0093-3} {\bibfield  {journal} {\bibinfo  {journal} {Sel.
  Math. New Ser.}\ }\textbf {\bibinfo {volume} {19}},\ \bibinfo {pages}
  {237--269} (\bibinfo {year} {2013}{\natexlab{b}})},\ \Eprint
  {http://arxiv.org/abs/1109.5558} {arXiv:1109.5558} \BibitemShut {NoStop}%
\bibitem [{\citenamefont {Aigner}\ and\ \citenamefont
  {Ziegler}(2001)}]{BookProof}%
  \BibitemOpen
  \bibfield  {author} {\bibinfo {author} {\bibfnamefont {Martin}\ \bibnamefont
  {Aigner}}\ and\ \bibinfo {author} {\bibfnamefont {G\"unter~M.}\ \bibnamefont
  {Ziegler}},\ }\href@noop {} {\emph {\bibinfo {title} {Proofs from {THE
  BOOK}}}},\ \bibinfo {edition} {2nd}\ ed.\ (\bibinfo  {publisher}
  {Springer-Verlag},\ \bibinfo {year} {2001})\BibitemShut {NoStop}%
\bibitem [{\citenamefont {Grunewald}\ \emph {et~al.}(1991)\citenamefont
  {Grunewald}, \citenamefont {Mennicke},\ and\ \citenamefont
  {Vaserstein}}]{GrunewaldMennickeVaserstein1991SymplecticGroup}%
  \BibitemOpen
  \bibfield  {author} {\bibinfo {author} {\bibfnamefont {Fritz}\ \bibnamefont
  {Grunewald}}, \bibinfo {author} {\bibfnamefont {Jens}\ \bibnamefont
  {Mennicke}}, \ and\ \bibinfo {author} {\bibfnamefont {Leonid}\ \bibnamefont
  {Vaserstein}},\ }\bibfield  {title} {\enquote {\bibinfo {title} {On
  symplectic groups over polynomial rings},}\ }\href {\doibase
  10.1007/BF02571323} {\bibfield  {journal} {\bibinfo  {journal} {Mathematische
  Zeitschrift}\ }\textbf {\bibinfo {volume} {206}},\ \bibinfo {pages} {35--56}
  (\bibinfo {year} {1991})}\BibitemShut {NoStop}%
\bibitem [{\citenamefont {Kopeyko}(1999)}]{Kopeyko1999Symplectic}%
  \BibitemOpen
  \bibfield  {author} {\bibinfo {author} {\bibfnamefont {V.~I.}\ \bibnamefont
  {Kopeyko}},\ }\bibfield  {title} {\enquote {\bibinfo {title} {Symplectic
  groups over {L}aurent polynomials, and patching diagrams},}\ }\href
  {http://mi.mathnet.ru/eng/fpm409} {\bibfield  {journal} {\bibinfo  {journal}
  {Fundam. Prikl. Mat.}\ }\textbf {\bibinfo {volume} {5}},\ \bibinfo {pages}
  {943--945} (\bibinfo {year} {1999})}\BibitemShut {NoStop}%
\bibitem [{\citenamefont {Walker}\ and\ \citenamefont
  {Wang}(2011)}]{Walker_2011}%
  \BibitemOpen
  \bibfield  {author} {\bibinfo {author} {\bibfnamefont {Kevin}\ \bibnamefont
  {Walker}}\ and\ \bibinfo {author} {\bibfnamefont {Zhenghan}\ \bibnamefont
  {Wang}},\ }\bibfield  {title} {\enquote {\bibinfo {title} {(3+1)-{TQFTs} and
  topological insulators},}\ }\href {\doibase 10.1007/s11467-011-0194-z}
  {\bibfield  {journal} {\bibinfo  {journal} {Frontiers of Physics}\ }\textbf
  {\bibinfo {volume} {7}},\ \bibinfo {pages} {150--159} (\bibinfo {year}
  {2011})},\ \Eprint {http://arxiv.org/abs/1104.2632} {arXiv:1104.2632}
  \BibitemShut {NoStop}%
\bibitem [{\citenamefont {Burnell}\ \emph {et~al.}(2014)\citenamefont
  {Burnell}, \citenamefont {Chen}, \citenamefont {Fidkowski},\ and\
  \citenamefont {Vishwanath}}]{BCFV}%
  \BibitemOpen
  \bibfield  {author} {\bibinfo {author} {\bibfnamefont {F.~J.}\ \bibnamefont
  {Burnell}}, \bibinfo {author} {\bibfnamefont {Xie}\ \bibnamefont {Chen}},
  \bibinfo {author} {\bibfnamefont {Lukasz}\ \bibnamefont {Fidkowski}}, \ and\
  \bibinfo {author} {\bibfnamefont {Ashvin}\ \bibnamefont {Vishwanath}},\
  }\bibfield  {title} {\enquote {\bibinfo {title} {Exactly soluble model of a
  {3D} symmetry protected topological phase of bosons with surface topological
  order},}\ }\href {\doibase 10.1103/PhysRevB.90.245122} {\bibfield  {journal}
  {\bibinfo  {journal} {Phys. Rev. B}\ }\textbf {\bibinfo {volume} {90}},\
  \bibinfo {pages} {245122} (\bibinfo {year} {2014})},\ \Eprint
  {http://arxiv.org/abs/1302.7072} {arXiv:1302.7072} \BibitemShut {NoStop}%
\end{thebibliography}%
\end{document}